\documentclass[notitlepage,nofootinbib,aps,prd,a4paper,11pt,amsfonts,amssymb,amsmath]{revtex4-1}
\usepackage{amsthm}
\usepackage[all]{xy}

\theoremstyle{plain}
\newtheorem{prop}{Proposition}
\theoremstyle{definition}
\newtheorem{defi}{Definition}

\newcounter{ccounter}

\newcounter{cpcounter}

\newcounter{fcounter}

\DeclareMathOperator{\Lie}{Lie}
\DeclareMathOperator{\Diff}{Diff}
\DeclareMathOperator{\Vect}{Vect}

\begin{document}
\title{Spacetime and observer space symmetries in the language of Cartan geometry}

\author{Manuel Hohmann}
\email{manuel.hohmann@ut.ee}
\affiliation{Teoreetilise F\"u\"usika Labor, F\"u\"usika Instituut, Tartu \"Ulikool, Ravila 14c, 50411 Tartu, Estonia}

\begin{abstract}
We introduce a definition of symmetry generating vector fields on manifolds which are equipped with a first-order reductive Cartan geometry. We apply this definition to a number of physically motivated examples and show that our newly introduced notion of symmetry agrees with the usual notions of symmetry of affine, Riemann-Cartan, Riemannian and Weizenb\"ock geometries, which are conventionally used as spacetime models. Further, we discuss the case of Cartan geometries which can be used to model observer space instead of spacetime. We show which vector fields on an observer space can be interpreted as symmetry generators of an underlying spacetime manifold, and may hence be called ``spatio-temporal''. We finally apply this construction to Finsler spacetimes and show that symmetry generating vector fields on a Finsler spacetime are indeed in a one-to-one correspondence with spatio-temporal vector fields on its observer space.
\end{abstract}
\maketitle

\section{Motivation}\label{sec:motivation}
In physics one often encounters the situation that a particular geometric object defined on a manifold is invariant under the action of a Lie group on that manifold. Such a Lie group action leaving some object invariant is usually called a \emph{symmetry} of that object. There are various examples of such objects, such as connections, metrics, tetrads, more general tensor fields or Finsler length measures. Each of these objects comes with its own definition of invariance under Lie group actions. The aim of this article is to provide a unified description of symmetries for a number of these objects. For this purpose we introduce a notion of symmetry for a particular structure known as Cartan geometry~\cite{Cartan:1935,Sharpe:1997}. It has the advantage that many different geometric objects induce particular subclasses of Cartan geometries. This allows us to describe their symmetries using a Cartan geometric framework.

One of the most prominent applications of geometry is the theory of gravity. Different theories of gravity are based on very different geometric objects, such as an affine connection in affine gravity~\cite{Einstein:1923a,Einstein:1923b,Eddington:1924}, a metric and a torsion in Einstein-Cartan gravity~\cite{Cartan:1922,Cartan:1923zea,Kibble:1961ba,Sciama:1964wt,Hehl:1976kj,Hehl:1994ue,Blagojevic:2002du,Blagojevic:2013xpa,Trautman:2006fp}, only a metric in general relativity and related theories~\cite{Einstein:1915ca,Carroll:2004}, or a tetrad in teleparallel gravity~\cite{Blagojevic:2002du,Blagojevic:2013xpa,Einstein:1928,Aldrovandi:2013wha,Baez:2012bn}. These objects have in common that they are sections of or connections on a fiber bundle over a spacetime manifold. However, also more general objects are used. An example is the Finsler length function~\cite{Chern:2000,Bucataru:2007}, which is used in Finsler geometric theories of gravity~\cite{Horvath:1950,Vacaru:2010fi,Vacaru:2007ng,Vacaru:2002kp,Pfeifer:2011tk,Pfeifer:2011xi,Pfeifer:2013gha}. Further, also Cartan geometry itself has been suggested as the mathematical framework for describing gravity and providing a possible link between canonical and covariant quantum gravity approaches~\cite{Wise:2006sm,Wise:2009fu,Gielen:2011mk,Gielen:2012fz,Wise:2013dja}. We will not discuss the physical content of these theories in this article, but keep our focus on the geometries they employ in order to model the gravitational interaction.

Knowing how to describe the symmetry of the geometric objects used in a gravity theory allows to classify solutions of its field equations by their degree of symmetry, as it is done, e.g., in the case of general relativity~\cite{Stephani:2003tm}. This task can be generalized by performing this classification for more general geometric objects, of which the aforementioned objects are special cases. While this may seem simple for geometric objects which are defined on spacetime itself, it is less obvious for objects such as Finsler length functions, which are defined on a bundle over spacetime. However, all of these objects can be described in terms of Cartan geometry, which is the reason for choosing the latter as the central object of the work presented in this article.

In fact, using Cartan geometry it is also possible to describe geometric models in which the underlying spacetime is not primary part of the description, or an absolute notion of spacetime which is independent of an observer may not even exist. This notion of a relative spacetime appears in a class of theories which describe physics on the space of observers, identified with the space of physical four-velocities, instead of spacetime~\cite{Gielen:2012fz,Wise:2013dja,Hohmann:2013fca,Hohmann:2014gpa}. Of course, also in this case we are interested in those symmetries which can be interpreted by an observer as transformations of (possibly relative) spacetime, and leave notions such as the relation between positions and velocities invariant. It is another aim of this article to provide a definition of such symmetries, which we call spacetime preserving, or ``spatio-temporal''.

Naturally the question arises why observers should be characterized only by their positions and four-velocities, and not by higher derivatives of their trajectories, such as their acceleration - in particular, since accelerated observers are distinguished from non-accelerated observers, e.g., in general relativity. The main reason why we restrict ourselves to positions and velocities in this work is because we wish to identify the trajectories of observers with those of test particles. For a test particle with given mass and gauge charge(s), the trajectory is commonly determined by a set of equations of motion derived from a Lagrangian. In order to avoid Ostrogradski instabilities, these should be of second order in their time derivatives~\cite{Ostrogradski:1850}. Hence, the acceleration and all higher derivatives of the observer's trajectory are fully determined from the position and velocity by the equations of motion. However, we remark that the notion of symmetry developed in this article could in principle be generalized to higher order observer spaces using jet bundles~\cite{Saunders:1989,Giachetta:2009}.

Of course, notions of symmetry for Cartan geometry already exist. The main difference between these standard notions of symmetry and the approach detailed in this article can be briefly summarized as follows. First, note that a Cartan geometry on a manifold \(M\) is defined in terms of a Lie algebra valued one-form \(A\), called the Cartan connection, on the total space \(P\) of a principal fiber bundle \(\pi: P \to M\). Morphisms between Cartan geometries \((\pi: P \to M, A)\) and \((\pi': P' \to M', A')\) are conventionally defined as principal bundle morphisms \(\Phi: P \to P'\) which preserve the Cartan connection, \(A = \Phi^*A'\)~\cite{CapSlovak:2009}, or as generalizations thereof involving a change of the Lie algebra~\cite{McKay:2008}. Symmetries of a Cartan geometry are thus modeled by principal bundle automorphisms, or infinitesimally by vector fields on the total space \(P\) of a Cartan geometry. Being the generators of bundle morphisms, these vector fields descend to vector fields on the base manifold \(M\).

The approach we discuss here proceeds in the opposite direction: we define infinitesimal symmetries of Cartan geometries in terms of vector fields on the base manifold \(M\), for which we provide a lifting prescription to the total space \(P\), obtained from their canonical lift to the general linear frame bundle \(\mathrm{GL}(M)\). This approach appears more natural in the context of symmetries of other spacetime geometries, which are usually modeled by vector fields on spacetime itself, and not on a bundle over spacetime.

This article is organized as follows. We start with a brief discussion of mathematical preliminaries in section~\ref{sec:preliminaries}. These are used for our definition of symmetries of general first-order reductive Cartan geometries in section~\ref{sec:cartan}. We then discuss several examples for model Klein geometries. In particular, we discuss affine Cartan geometry in section~\ref{sec:affine} and orthogonal Cartan geometry in section~\ref{sec:orthogonal}, both of which are used as models for spacetime geometries. From there we move on to the geometry of general observer spaces, whose symmetries we discuss in section~\ref{sec:observer}. As a particular example we discuss orthogonal observer space geometry in section~\ref{sec:ortobs}, which in particular includes Finsler spacetime geometry. We end with a conclusion in section~\ref{sec:conclusion}.

\section{Preliminaries}\label{sec:preliminaries}
In this section we briefly review the mathematical preliminaries needed for our constructions during the remainder of this article. In section~\ref{subsec:lie} we discuss the Lie derivative and in particular explain how it can be applied to affine connections. We then show how vector fields on a manifold can be lifted to its tangent bundle and general linear frame bundle in section~\ref{subsec:lifts}. We finally display the definitions of the geometries we will be working with, which are Cartan geometry in section~\ref{subsec:cartan} and Finsler spacetime geometry in section~\ref{subsec:finsler}.

\subsection{The Lie derivative}\label{subsec:lie}
Symmetries of spacetimes are conventionally described by the invariance of geometrical objects, such as the metric or the torsion, under a finite-dimensional subgroup \(S\) of the diffeomorphism group \(\Diff(M)\) acting on the spacetime manifold \(M\), which carries the structure of a Lie group. The Lie algebra \(\mathfrak{s} = \Lie S\) of this symmetry group is then given by a subalgebra of the algebra \(\Vect(M)\) of vector fields on \(M\). It is well known that the action of a diffeomorphism generated by the flow of a vector field \(\xi\) on a tensor field \(W\) is described by the Lie derivative \(\mathcal{L}_{\xi}W\). However, it is less common practice to apply the Lie derivative to objects which are not tensor fields on \(M\), in particular to affine connections. For this reason we briefly summarize the definition and basic formulas for the Lie derivative of connections in this section.

Let \(\varphi: \mathbb{R} \times M \to M\) be a one-parameter group of diffeomorphisms on \(M\) generated by a vector field \(\xi \in \Vect(M)\). This means that for all \(t \in \mathbb{R}\), \(\varphi_t: M \to M\) is the unique diffeomorphism on \(M\) such that for all \(x \in M\) the tangent vectors to the curve \(t \mapsto \varphi_t(x)\) are given by
\begin{equation}
\frac{d}{dt}\varphi_t(x) = \xi(\varphi_t(x))
\end{equation}
and \(\varphi_0 = \mathrm{id}_M\). We call \(\varphi_t\) the flow of \(\xi\). Note that \(\varphi_t \circ \varphi_{t'} = \varphi_{t + t'}\) and in particular \(\varphi_{-t} = \varphi_t^{-1}\). For a tensor field \(W\) of type \((m,n)\), which is a section of the bundle \(TM^{\otimes m} \otimes T^*M^{\otimes n}\), we define the pullback by \(\varphi_t\) as
\begin{equation}\label{eqn:tenspullback}
\varphi_t^*W = W \circ (\underbrace{\varphi_{-t*} \otimes \ldots \otimes \varphi_{-t*}}_{m \text{ times}} \otimes \underbrace{\varphi_{t*} \otimes \ldots \otimes \varphi_{t*}}_{n \text{ times}})
\end{equation}
where \(\varphi_{t*}: TM \to TM\) denotes the differential of \(\varphi_t\). Using coordinates \((x^a)\) and writing \(x' = \varphi_t(x)\) this yields the familiar relation
\begin{equation}
(\varphi_t^*W)^{a_1 \cdots a_m}{}_{b_1 \cdots b_n}(x) = W^{c_1 \cdots c_m}{}_{d_1 \cdots d_n}(x')\frac{\partial x^{a_1}}{\partial x'^{c_1}}\ldots\frac{\partial x^{a_m}}{\partial x'^{c_m}}\frac{\partial x'^{d_1}}{\partial x^{b_1}}\ldots\frac{\partial x'^{d_n}}{\partial x^{b_n}}\,,
\end{equation}
which describes the transformation of the components of \(W\) under diffeomorphisms. The Lie derivative is then defined as
\begin{equation}\label{eqn:tensliedef}
(\mathcal{L}_{\xi}W) = \left.\frac{d}{dt}(\varphi_t^*W)\right|_{t = 0}\,.
\end{equation}
It is easy to see that in coordinates it yields the standard formula
\begin{equation}
\begin{split}
(\mathcal{L}_{\xi}W)^{a_1 \cdots a_m}{}_{b_1 \cdots b_n} &= \xi^c\partial_cW^{a_1 \cdots a_m}{}_{b_1 \cdots b_n}\\
&\phantom{=}- \partial_c\xi^{a_1}W^{ca_2 \cdots a_m}{}_{b_1 \cdots b_n} - \ldots - \partial_c\xi^{a_m}W^{a_1 \cdots a_{m - 1}c}{}_{b_1 \cdots b_n}\\
&\phantom{=}+ \partial_{b_1}\xi^cW^{a_1 \cdots a_m}{}_{cb_2 \cdots b_n} + \ldots + \partial_{b_n}\xi^cW^{a_1 \cdots a_m}{}_{b_1 \cdots b_{n - 1}c}\,.
\end{split}
\end{equation}
Note that \(\mathcal{L}_{\xi}W\) is a tensor with the same shape as \(W\).

Given an affine connection \(\nabla\) on \(M\) we define the pullback \(\varphi_t^*\nabla\) similarly to the tensor case~\eqref{eqn:tenspullback} as
\begin{equation}
(\varphi_t^*\nabla)_XY = \varphi_{-t*} \circ \nabla_{\varphi_{t*} \circ X}(\varphi_{t*} \circ Y)\,,
\end{equation}
where \(X, Y\) are vector fields. From the definition of the connection coefficients\footnote{Note that there exist different conventions in the literature regarding the order of the lower two indices. Here we choose the second lower index to be the ``derivative'' index.}
\begin{equation}
\Gamma^a{}_{bc} = dx^a(\nabla_{\partial_c}\partial_b)
\end{equation}
one then directly reads off
\begin{equation}
(\varphi_t^*\Gamma)^a{}_{bc}(x) = \Gamma^d{}_{ef}(x')\frac{\partial x^a}{\partial x'^d}\frac{\partial x'^e}{\partial x^b}\frac{\partial x'^f}{\partial x^c} + \frac{\partial x^a}{\partial x'^d}\frac{\partial^2 x'^d}{\partial x^b \partial x^c}\,,
\end{equation}
which is the familiar formula for the transformation of connection coefficients under a diffeomorphism. Formally, we can then define the Lie derivative of the connection as
\begin{equation}
(\mathcal{L}_{\xi}\nabla) = \left.\frac{d}{dt}(\varphi_t^*\nabla)\right|_{t = 0}\,,
\end{equation}
whose meaning becomes clear by writing out the connection components
\begin{equation}\label{eqn:connlieder}
\begin{split}
(\mathcal{L}_{\xi}\Gamma)^a{}_{bc} &= \xi^d\partial_d\Gamma^a{}_{bc} - \partial_d\xi^a\Gamma^d{}_{bc} + \partial_b\xi^d\Gamma^a{}_{dc} + \partial_c\xi^d\Gamma^a{}_{bd} + \partial_b\partial_c\xi^a\\
&= \nabla_c\nabla_b\xi^a - \xi^dR^a{}_{bcd} - \nabla_c(\xi^dT^a{}_{bd})\,,
\end{split}
\end{equation}
where \(T\) and \(R\) denote the torsion and curvature tensors of \(\nabla\), respectively, and we use the sign conventions
\begin{equation}\label{eqn:riemann}
R^a{}_{bcd} = \partial_c\Gamma^a{}_{bd} - \partial_d\Gamma^a{}_{bc} + \Gamma^a{}_{ec}\Gamma^e{}_{bd} - \Gamma^a{}_{ed}\Gamma^e{}_{bc}
\end{equation}
and
\begin{equation}\label{eqn:torsion}
T^a{}_{bc} = \Gamma^a{}_{cb} - \Gamma^a{}_{bc}\,.
\end{equation}
Note that \(\mathcal{L}_{\xi}\Gamma\) is a tensor field, which can immediately be seen from the expression in the second line in equation~\eqref{eqn:connlieder}.

One can also define the Lie derivatives of geometric objects which are defined not on the manifold itself, but on its tangent or frame bundle. This will be shown in the next section and in section~\ref{subsec:finsler}.

\subsection{Complete lifts of vector fields}\label{subsec:lifts}
In this article we wish to express symmetries in terms of vector fields on a manifold \(M\) as summarized in the previous section. However, the objects which should be invariant under the flow of this vector field will be defined not on \(M\) itself, but on the tangent bundle \(TM\) or (a subbundle of) the bundle \(\mathrm{GL}(M)\) of linear frames of \(TM\), as defined, e.g., in~\cite{Kobayashi:1963}. Since these bundles are examples of \emph{natural bundles}, there exists a \emph{functorial lift} of smooth maps, and thus in particular (infinitesimal) diffeomorphisms, to smooth bundle morphisms, and thus in particular (infinitesimal) bundle automorphisms on these bundles~\cite{Kolar:1993}. In this section we briefly review how this fact allows lifting vector fields from the base manifold to the total space of these bundles.

We start with the complete lift of a vector field to the tangent bundle, following the construction in~\cite{Yano:1973}: Let \(\varphi: \mathbb{R} \times M \to M\) be a one-parameter group of diffeomorphisms on \(M\) generated by a vector field \(\xi\) on \(M\) as in the preceding section. For each \(t \in \mathbb{R}\) the differential \(\varphi_{t*}: TM \to TM\) of \(\varphi_t: M \to M\) is a vector bundle automorphism of \(TM\). One easily checks that the definition \(\hat{\varphi}_t = \varphi_{t*}\) yields a one-parameter group \(\hat{\varphi}\) of diffeomorphisms on \(TM\). The vector field \(\hat{\xi}\) on \(TM\) generating \(\hat{\varphi}\) is called the complete lift of \(\xi\).

It is convenient to introduce coordinates on \(TM\) as follows. Given coordinates \((x^a)\) on \(M\), we can denote any tangent vector by
\begin{equation}
y^a\left.\frac{\partial}{\partial x^a}\right|_x
\end{equation}
in the coordinate basis \(\partial/\partial x^a\) of \(TM\). The coordinates \((x^a,y^a)\) on \(TM\) are called induced coordinates and will be used throughout this article. For convenience we further introduce the notation
\begin{equation}
\left\{\partial_a = \frac{\partial}{\partial x^a}, \bar{\partial}_a = \frac{\partial}{\partial y^a}\right\}
\end{equation}
for the induced coordinate basis of \(TTM\). Using this notation the complete lift \(\hat{\xi}\) of \(\xi\) to \(TM\) can be written in the form
\begin{equation}
\hat{\xi} = \xi^a\partial_a + y^b\partial_b\xi^a\bar{\partial}_a\,.
\end{equation}

We now come to complete lifts of vector fields to the general linear frame bundle \(\mathrm{GL}(M)\), which can be defined as
\begin{equation}\label{eqn:glframebundle}
\mathrm{GL}(M) = \bigcup_{x \in M}\{\text{linear bijections } f: \mathbb{R}^n \to T_xM\}\,,
\end{equation}
where \(n\) is the dimension of \(M\). A diffeomorphism \(\varphi_t\) on \(M\) then defines a diffeomorphism \(\bar{\varphi}_t\) on \(\mathrm{GL}(M)\) as \(\bar{\varphi}_t(f) = \varphi_{t*} \circ f\). A quick calculation shows that this definition yields a one-parameter group \(\bar{\varphi}\) of principal bundle automorphisms of \(\mathrm{GL}(M)\) which is generated by a vector field \(\bar{\xi}\). We call \(\bar{\xi}\) the complete lift of \(\xi\) to the frame bundle.

On the frame bundle we can introduce coordinates similar to the induced coordinates on the tangent bundle. Given coordinates \((x^a)\) on \(M\), we can denote any linear map \(f: \mathbb{R}^n \to T_xM\) by
\begin{equation}
f_i^a\left.\frac{\partial}{\partial x^a}\right|_x\,,
\end{equation}
where the index \(i\) refers to the canonical basis of \(\mathbb{R}^n\). We also call \((x^a,f_i^a)\) induced coordinates on \(\mathrm{GL}(M)\). For the induced coordinate basis of \(T\mathrm{GL}(M)\) we introduce the notation
\begin{equation}
\left\{\partial_a = \frac{\partial}{\partial x^a}, \bar{\partial}^i_a = \frac{\partial}{\partial f_i^a}\right\}\,.
\end{equation}
In this notation the complete lift \(\bar{\xi}\) of \(\xi\) to \(\mathrm{GL}(M)\) reads
\begin{equation}\label{eqn:glclift}
\bar{\xi} = \xi^a\partial_a + f_i^b\partial_b\xi^a\bar{\partial}^i_a\,.
\end{equation}
Note that we could have obtained this expression also using the following definition: Let \(\theta \in \Omega^1(\mathrm{GL}(M), \mathbb{R}^n)\) be the solder (or canonical) form on \(M\)~\cite{Kobayashi:1963}, which in our coordinates takes the form \(\theta^i = f^{-1}{}^i_adx^a\) in the canonical basis of \(\mathbb{R}^n\). Then the complete lift \(\bar{\xi}\) is the unique vector field on \(\mathrm{GL}(M)\) which projects to \(\xi\) on \(M\) and leaves the canonical form invariant, \(\mathcal{L}_{\bar{\xi}}\theta = 0\). One easily checks that also this definition yields the coordinate expression~\eqref{eqn:glclift}.

We will frequently make use of the expressions for the complete lifts in the remainder of this article.

\subsection{Cartan geometry}\label{subsec:cartan}
The main geometric structure we use in this article is Cartan geometry as developed in~\cite{Cartan:1935}; see~\cite{Sharpe:1997,Wise:2006sm} for a detailed discussion. The Cartan geometry of a manifold \(M\) is based on a comparison to the geometry of a suitable model space, which is a homogeneous space (in this context also called Klein geometry), i.e., the coset space \(G/H\) of a Lie group \(G\) and a closed subgroup \(H \subset G\). For a given a homogeneous space, a Cartan geometry can be defined as follows:

\begin{defi}[Cartan geometry]\label{def:cartan}
A \emph{Cartan geometry} modeled on a Klein geometry \(G/H\) is a principal $H$-bundle \(\pi: \mathcal{P} \to M\) together with a $\mathfrak{g}$-valued 1-form \(A \in \Omega^1(\mathcal{P},\mathfrak{g})\) on \(\mathcal{P}\), such that
\begin{list}{\textrm{C\arabic{ccounter}}.}{\usecounter{ccounter}}
\item\label{cartan:isomorphism}
For each \(p \in \mathcal{P}\), \(A_p: T_p\mathcal{P} \to \mathfrak{g}\) is a linear isomorphism.
\item\label{cartan:equivariant}
\(A\) is $H$-equivariant: \((R_h)^*A = \mathrm{Ad}(h^{-1}) \circ A\) \(\forall h \in H\).
\item\label{cartan:mcform}
\(A(\tilde{h}) = h\) for all \(h \in \mathfrak{h}\), where \(\tilde{h}\) denotes the fundamental vector field of \(h\).
\end{list}
\end{defi}

In the definition above, the fundamental vector field \(\tilde{h} \in \Vect(\mathcal{P})\) associated to \(h \in \mathfrak{h}\) via the action of the group \(H\) on \(\mathcal{P}\) is the vector field generating a one-parameter group of (vertical) diffeomorphisms of \(\mathcal{P}\), which corresponds to the action of the one-parameter subgroup of \(H\) generated by \(h\) on \(\mathcal{P}\)~\cite{Kobayashi:1963}.

Instead of describing the Cartan geometry in terms of the Cartan connection \(A\), which is equivalent to specifying a linear isomorphism \(A_p: T_p\mathcal{P} \to \mathfrak{g}\) for all \(p \in \mathcal{P}\) due to condition~\ref{cartan:isomorphism}, we can use the inverse maps \(\underline{A}_p = A_p^{-1}: \mathfrak{g} \to T_p\mathcal{P}\). For each \(a \in \mathfrak{g}\) they define a section \(\underline{A}(a)\) of the tangent bundle \(T\mathcal{P}\), which we call an associated vector field:

\begin{defi}[Associated vector fields]
Let \((\pi: \mathcal{P} \to M, A)\) be a Cartan geometry modeled on \(G/H\). For each \(a \in \mathfrak{g}\) the \emph{associated vector field} \(\underline{A}(a)\) is the unique vector field such that \(A(\underline{A}(a)) = a\).
\end{defi}

We can therefore equivalently define a Cartan geometry in terms of its associated vector fields, due to the following proposition:

\begin{prop}
Let \((\pi: \mathcal{P} \to M, A)\) be a Cartan geometry modeled on \(G/H\) and \(\underline{A}: \mathfrak{g} \to \Vect(\mathcal{P})\) its associated vector fields. Then the properties~\ref{cartan:isomorphism} to~\ref{cartan:mcform} of \(A\) are respectively equivalent to the following properties of \(\underline{A}\):
\begin{list}{\textrm{C\arabic{cpcounter}}'.}{\usecounter{cpcounter}}
\item\label{cartan:isomorphism2}
For each \(p \in \mathcal{P}\), \(\underline{A}_p: \mathfrak{g} \to T_p\mathcal{P}\) is a linear isomorphism.
\item\label{cartan:equivariant2}
\(\underline{A}\) is $H$-equivariant: \(R_{h*} \circ \underline{A} = \underline{A} \circ \mathrm{Ad}(h^{-1})\) \(\forall h \in H\).
\item\label{cartan:canonical}
\(\underline{A}\) restricts to the fundamental vector fields on \(\mathfrak{h}\): \(\underline{A}(h) = \tilde{h}\) \(\forall h \in \mathfrak{h}\).
\end{list}
\end{prop}

We add a few remarks to the definition of Cartan geometry given above. Recall that for each \(x \in M\) the fiber \(\mathcal{P}_x\) over \(x\) can be identified with the group \(H\), up to left multiplication with an arbitrary element of \(H\). Via this identification, the Maurer-Cartan form \(\mu \in \Omega^1(H, \mathfrak{h})\) induces a one-form \(\mu_x \in \Omega^1(\mathcal{P}_x, \mathfrak{h})\) on each fiber. Condition~\ref{cartan:mcform} then states that the restriction of \(A\) to each fiber \(\mathcal{P}_x\) agrees with the induced one-form \(\mu_x\)~\cite{Sharpe:1997}. In a similar fashion, the fundamental vector fields \(\tilde{h}\) for \(h \in \mathfrak{h}\) are related to the left invariant vector fields on \(H\) via the same identification. Finally, if there exists a split \(\mathfrak{g} = \mathfrak{h} \oplus \mathfrak{z}\) into (not necessarily irreducible) subrepresentations of the adjoint representation of \(H \subset G\) on \(\mathfrak{g}\), and hence a split \(A = \omega + e\) of the Cartan connection, one can identify the associated vector fields \(\underline{A}(\mathfrak{z})\) with the standard horizontal vector fields, which lie in the kernel of the $\mathfrak{h}$-valued part \(\omega\) of the Cartan connection (which is a principal bundle connection)~\cite{Kobayashi:1963}.

Further, for \(h \in \mathfrak{h}\) and \(g \in \mathfrak{g}\) it follows from condition~\ref{cartan:equivariant2} that
\begin{equation}\label{eqn:equivariant2}
\mathcal{L}_{\tilde{h}}\underline{A}(g) = \underline{A}([h,g])\,.
\end{equation}
Similarly, condition~\ref{cartan:equivariant} implies that
\begin{equation}\label{eqn:equivariant}
\mathcal{L}_{\tilde{h}}A = -[h,A]\,.
\end{equation}
Note that the last two equations are equivalent to the original conditions~\ref{cartan:equivariant2} and~\ref{cartan:equivariant} if and only if \(H\) is connected. We will make use of these formulas when we discuss particular Cartan geometries in order to restrict the allowed Cartan connections.

In this article we will frequently encounter (infinitesimal generators of) maps between Cartan geometries. In the context of symmetry we will need a notion of maps which preserve the Cartan geometry. These maps are defined as follows~\cite{CapSlovak:2009}.
\begin{defi}[Cartan geometry morphism]
A \emph{morphism of Cartan geometries} between \((\pi: \mathcal{P} \to M, A)\) and \((\pi': \mathcal{P}' \to M', A')\) modeled on \(G/H\) is a principal bundle morphism \(\Phi: \mathcal{P} \to \mathcal{P}'\) such that \(A = \Phi^*A'\).
\end{defi}
In this article we will deal mostly with automorphisms of Cartan geometries, and one-parameter groups of automorphisms and their generating vector fields in particular. We will return to this topic in section~\ref{sec:cartan}, when we discuss symmetries of Cartan geometries in more detail.

\subsection{Finsler spacetimes}\label{subsec:finsler}
A particularly interesting application for the framework discussed in this article, which goes beyond the classical description of spacetime geometries in terms on metrics and connections, is the application to Finsler geometry~\cite{Bucataru:2007}. Here we will employ the Finsler spacetime framework developed in~\cite{Pfeifer:2011tk,Pfeifer:2011xi,Pfeifer:2013gha}. Its central definition is that of a Finsler spacetime, which we use as formulated in~\cite{Hohmann:2013fca,Hohmann:2014gpa}, and which we briefly review here, as it will be crucial for our construction in section~\ref{subsec:finslersym}:

\begin{defi}[Finsler spacetime]
A \emph{Finsler spacetime} \((M,L,F)\) of dimension \(n\) is a $n$-dimensional, connected, Hausdorff, paracompact, smooth manifold \(M\) equipped with continuous real functions \(L, F\) on the tangent bundle \(TM\) which has the following properties:
\begin{list}{\textrm{F\arabic{fcounter}}.}{\usecounter{fcounter}}
\item\label{finsler:lsmooth}
\(L\) is smooth on the tangent bundle without the zero section \(TM \setminus \{0\}\).
\item\label{finsler:lhomogeneous}
\(L\) is positively homogeneous of real degree \(r \geq 2\) with respect to the fiber coordinates of~\(TM\),
\begin{equation}
L(x,\lambda y) = \lambda^rL(x,y) \quad \forall \lambda > 0\,,
\end{equation}
and defines the Finsler function \(F\) via \(F(x,y) = |L(x,y)|^{\frac{1}{r}}\).
\item\label{finsler:lreversible}
\(L\) is reversible: \(|L(x,-y)| = |L(x,y)|\).
\item\label{finsler:lhessian}
The Hessian
\begin{equation}
g^L_{ab}(x,y) = \frac{1}{2}\bar{\partial}_a\bar{\partial}_bL(x,y)
\end{equation}
of \(L\) with respect to the fiber coordinates is non-degenerate on \(TM \setminus X\), where \(X \subset TM\) has measure zero and does not contain the null set \(\{(x,y) \in TM | L(x,y) = 0\}\).
\item\label{finsler:timelike}
The unit timelike condition holds, i.e., for all \(x \in M\) the set
\begin{equation}
\Omega_x = \left\{y \in T_xM \left| |L(x,y)| = 1, g^L_{ab}(x,y) \text{ has signature } (\epsilon,-\epsilon,\ldots,-\epsilon), \epsilon = \frac{L(x,y)}{|L(x,y)|}\right.\right\}
\end{equation}
contains a non-empty closed connected component~\({S_x \subseteq \Omega_x \subset T_xM}\).
\end{list}
\end{defi}
In contrast to the standard notion of Euclidean Finsler geometry, which is defined only by the Finsler function \(F\), this definition uses an auxiliary function \(L\) in order to define a suitable null structure via condition~\ref{finsler:lhessian} and timelike vectors via condition~\ref{finsler:timelike}. The main advantage of using this auxiliary function is the fact that in contrast to \(F\) it is differentiable also on the null structure, and can thus be used to define geometric objects such as null geodesics. In the following we will show the definitions of a few of these objects and discuss their behavior under diffeomorphisms of the underlying spacetime manifold.

An important object in Finsler geometry and related geometries on the tangent bundle is the \emph{Cartan non-linear connection}~\cite{Bucataru:2007}
\begin{equation}
\left.\nabla_X\partial_b\right|_x = \left.N^a{}_b(x, X(x))\partial_a\right|_x
\end{equation}
on the manifold \(M\), which is non-linear in the vector \(X\) specifying the direction of differentiation. Its components \(N^a{}_b(x,y)\) are functions on the tangent bundle \(TM\) and can be derived from the geometry function \(L\) as
\begin{equation}\label{eqn:nonlinconn}
N^a{}_b = \frac{1}{4}\bar{\partial}_b\left[g^{L\,ac}(y^d\partial_d\bar{\partial}_cL - \partial_cL)\right]\,.
\end{equation}
From the Cartan non-linear connection we can derive the \emph{Berwald basis}~\cite{Bucataru:2007}
\begin{equation}\label{eqn:finslerhorz}
\{\delta_a = \partial_a - N^b{}_a\bar{\partial}_b, \bar{\partial}_a\}
\end{equation}
of \(TTM\). Its dual basis of \(T^*TM\) is given by
\begin{equation}\label{eqn:cotberwald}
\{dx^a, \delta y^a = dy^a + N^a{}_bdx^b\}\,.
\end{equation}
In addition to the non-linear connection on \(M\), there exist a number of linear connections on \(TM\) which are compatible with the Cartan non-linear connection in the sense that they respect the split of the tangent bundle manifest by the Berwald basis~\eqref{eqn:finslerhorz} and operate identically on vector fields \(\xi^a\delta_a\) and \(\xi^a\bar{\partial}_a\). A connection satisfying these conditions is called \emph{$N$-linear connection} and is of the form
\begin{equation}
\nabla_{\delta_a}\delta_b = F^c{}_{ba}\delta_c\,, \quad \nabla_{\delta_a}\bar{\partial}_b = F^c{}_{ba}\bar{\partial}_c\,, \quad \nabla_{\bar{\partial}_a}\delta_b = C^c{}_{ba}\delta_c\,, \quad \nabla_{\bar{\partial}_a}\bar{\partial}_b = C^c{}_{ba}\bar{\partial}_c\,.
\end{equation}
The coefficients \(F^c{}_{ba}\) and \(C^c{}_{ba}\) are functions on the tangent bundle and are defined differently for the different types of linear connections; see~\cite{Chern:2000,Bucataru:2007} for a thorough discussion. Here we will in particular use the Cartan linear connection, whose coefficients are given by
\begin{equation}\label{eqn:cartlinconn}
F^a{}_{bc} = \frac{1}{2}g^{F\,ad}(\delta_bg^F_{cd} + \delta_cg^F_{bd} - \delta_dg^F_{bc})\,, \quad
C^a{}_{bc} = \frac{1}{2}g^{F\,ad}(\bar{\partial}_bg^F_{cd} + \bar{\partial}_cg^F_{bd} - \bar{\partial}_dg^F_{bc})\,,
\end{equation}
where \(g^F_{ab} = \frac{1}{2}\bar{\partial}_a\bar{\partial}_bF^2\) is the Finsler metric. It has the property that the Finsler metric is covariantly constant with respect to the Cartan linear connection. This plays an important role for the relation between Finsler an Cartan geometry, which we use in section~\ref{subsec:finslercartan}.

We finally discuss the question under which conditions a Finsler spacetime is invariant under a diffeomorphism of the underlying spacetime manifold. The notion of symmetry we consider her has been extensively studied in~\cite{Tashiro:1959}; here we only display the parts which are relevant for our calculation in section~\ref{subsec:finslersym}.\footnote{Note that~\cite{Tashiro:1959} also discusses symmetries of Cartan spaces. However, these are different geometric structures than the Cartan geometries we reviewed in section~\ref{subsec:cartan}.} Recall from section~\ref{subsec:lifts} that a one-parameter subgroup \(\varphi\) of diffeomorphisms of \(M\) generated by a vector field \(\xi\) induces a one-parameter subgroup \(\hat{\varphi}\) of diffeomorphisms of \(TM\) generated by a vector field \(\hat{\xi}\), which we call the complete lift of \(\xi\) to \(TM\). Under these induced diffeomorphisms the Finsler geometric objects defined above transform as follows.
\begin{itemize}
\item
The geometry function \(L\) transforms as a scalar on \(TM\),
\begin{equation}
(\hat{\varphi}_t^*L)(x,y) = L(x',y')\,.
\end{equation}
Its Lie derivative is thus given by
\begin{equation}
\mathcal{L}_{\hat{\xi}}L = \xi^a\partial_aL + y^b\partial_b\xi^a\bar{\partial}_aL\,.
\end{equation}
\item
The Hessian \(g^L_{ab}\) and the Finsler metric \(g^F_{ab}\) transform like tensors on the manifold as
\begin{equation}
(\hat{\varphi}_t^*g^{L/F})_{ab}(x,y) = g^{L/F}_{cd}(x',y')\frac{\partial x'^c}{\partial x^a}\frac{\partial x'^d}{\partial x^b}\,.
\end{equation}
Objects with this transformation behavior are called d-tensors. The Lie derivative is given by
\begin{equation}
(\mathcal{L}_{\hat{\xi}}g^{L/F})_{ab} = \xi^c\partial_cg^{L/F}_{ab} + y^d\partial_d\xi^c\bar{\partial}_cg^{L/F}_{ab} + \partial_a\xi^cg^{L/F}_{cb} + \partial_b\xi^cg^{L/F}_{ac}\,.
\end{equation}
\item
The coefficients \(N^a{}_b\) of the Cartan non-linear connection transform as
\begin{equation}
(\hat{\varphi}_t^*N)^a{}_b(x,y) = N^c{}_d(x',y')\frac{\partial x^a}{\partial x'^c}\frac{\partial x'^d}{\partial x^b} + y^d\frac{\partial^2 x'^c}{\partial x^b\partial x^d}\frac{\partial x^a}{\partial x'^c}\,.
\end{equation}
From this follows the Lie derivative
\begin{equation}
(\mathcal{L}_{\hat{\xi}}N)^a{}_b = \xi^c\partial_cN^a{}_b + y^d\partial_d\xi^c\bar{\partial}_cN^a{}_b - \partial_c\xi^aN^c{}_b + \partial_b\xi^cN^a{}_c + y^c\partial_b\partial_c\xi^a\,.
\end{equation}
\item
The coefficients \(F^a{}_{bc}\) and \(C^a{}_{bc}\) of any $N$-linear connection transform as
\begin{subequations}
\begin{align}
(\hat{\varphi}^*F)^a{}_{bc}(x,y) &= F^d{}_{ef}(x',y')\frac{\partial x^a}{\partial x'^d}\frac{\partial x'^e}{\partial x^b}\frac{\partial x'^f}{\partial x^c} + \frac{\partial x^a}{\partial x'^d}\frac{\partial^2 x'^d}{\partial x^b \partial x^c}\,,\\
(\hat{\varphi}^*C)^a{}_{bc}(x,y) &= C^d{}_{ef}(x',y')\frac{\partial x^a}{\partial x'^d}\frac{\partial x'^e}{\partial x^b}\frac{\partial x'^f}{\partial x^c}\,,
\end{align}
\end{subequations}
so that also \(C^a{}_{bc}\) is a d-tensor. Their Lie derivative therefore takes the form
\begin{subequations}
\begin{align}
(\mathcal{L}_{\hat{\xi}}F)^a{}_{bc} &= \xi^d\partial_dF^a{}_{bc} + y^e\partial_e\xi^d\bar{\partial}_dF^a{}_{bc} - \partial_d\xi^aF^d{}_{bc} + \partial_b\xi^dF^a{}_{dc} + \partial_c\xi^dF^a{}_{bd} + \partial_b\partial_c\xi^a\,,\\
(\mathcal{L}_{\hat{\xi}}C)^a{}_{bc} &= \xi^d\partial_dC^a{}_{bc} + y^e\partial_e\xi^d\bar{\partial}_dC^a{}_{bc} - \partial_d\xi^aC^d{}_{bc} + \partial_b\xi^dC^a{}_{dc} + \partial_c\xi^dC^a{}_{bd}\,.
\end{align}
\end{subequations}
\end{itemize}
Note that all objects listed above are derived from the fundamental geometry function \(L\). A lengthy, but straightforward calculation shows that all of their Lie derivatives vanish for a given vector field \(\xi\) on \(M\) if \(\mathcal{L}_{\hat{\xi}}L = 0\). A vector field satisfying this condition is therefore regarded a symmetry of the Finsler spacetime. We will make use of this property when we discuss the symmetries of Finsler spacetimes in more detail in section~\ref{subsec:finslersym}.

This concludes our brief introduction of the basic mathematical notions we will be working with. We will continue with a deeper discussion of Cartan geometry as introduced in section~\ref{subsec:cartan}, in particular under the aspect of symmetry.

\section{Symmetries of first-order reductive Cartan geometries}\label{sec:cartan}
After clarifying the mathematical preliminaries we now come to the discussion of symmetries in the language of Cartan geometry. While we have given the general definition of a Cartan geometry in section~\ref{subsec:cartan}, we now restrict ourselves to a class of Cartan geometries called first-order reductive. We provide their definition and discuss a few of their properties in section~\ref{subsec:forcartan}. For these we then construct a notion of symmetry, i.e., invariance under the flow of a vector field, in section~\ref{subsec:cartansym}.

\subsection{First order reductive Cartan geometries}\label{subsec:forcartan}
In this article we will be dealing with Cartan geometries whose model Klein geometry has additional properties. Note that for any Klein geometry \(G/H\) the adjoint representation of \(G\) on its Lie algebra \(\mathfrak{g}\) yields a representation of \(H \subset G\) on \(\mathfrak{g}\), which has the adjoint representation of \(H\) on \(\mathfrak{h} \subset \mathfrak{g}\) as a subrepresentation. The following property of the quotient representation on \(\mathfrak{g}/\mathfrak{h}\), which is a vector space isomorphic to any tangent space of the homogeneous space \(G/H\), will be crucial for our constructions in this article~\cite{Sharpe:1997}:
\begin{defi}[First-order Cartan geometry]
A Cartan geometry with model Klein geometry \(G/H\) is called first-order Cartan geometry if the quotient representation of \(H\) on \(\mathfrak{g}/\mathfrak{h}\) is faithful. Otherwise, it is called higher-order Cartan geometry.
\end{defi}
In this article we consider only first-order Cartan geometries. Following the definition given in~\cite{Gielen:2012fz}, we now consider the associated bundle
\begin{equation}
\mathcal{T} = \mathcal{P} \times_H \mathfrak{g}/\mathfrak{h}\,,
\end{equation}
which is called the \emph{fake tangent bundle}~\cite{Wise:2006sm}. Its elements are equivalence classes \([p,z]\) of pairs \((p,z) \in \mathcal{P} \times \mathfrak{g}/\mathfrak{h}\) under the identification
\begin{equation}\label{eqn:assocbundle}
[p,z] = [R_hp, \mathrm{Ad}(h^{-1})z]\,, \quad h \in H\,.
\end{equation}
For \(x \in M\), we call a linear map \(\mathfrak{g}/\mathfrak{h} \to \mathcal{T}_x\) an \emph{admissible frame} of \(\mathcal{T}_x\) if it is of the form \(z \mapsto [p,z]\) for some \(p \in \mathcal{P}_x\). Note that for a first-order Cartan geometry there is a canonical one-to-one correspondence between elements of \(\mathcal{P}\) and admissible frames of the fake tangent bundle. For this reason, we can canonically identify \(\mathcal{P}\) with this space of admissible frames, which we call the \emph{fake frame bundle}, and write \(p(z)\) instead of \([p,z]\). This identification also allows us to write the action of \(H\) on \(\mathcal{P}\) as \(R_hp = p \circ \mathrm{Ad}(h)\), which immediately follows from equation~\eqref{eqn:assocbundle}.

Note that there is no canonical relation between the fake tangent bundle \(\mathcal{T}\) and the tangent bundle \(TM\). However, there is a distinguished one given by a \emph{coframe field}, which is a vector bundle isomorphism from \(TM\) to \(\mathcal{T}\). We can obtain a coframe field from the Cartan connection, provided that the following condition holds~\cite{Sharpe:1997}:
\begin{defi}[Reductive Cartan geometry]
A Cartan geometry with model Klein geometry \(G/H\) is called \emph{reductive} if the Lie algebra \(\mathfrak{g}\) allows a decomposition of the form \(\mathfrak{g} = \mathfrak{h} \oplus \mathfrak{z}\) into subrepresentations of the adjoint representation of \(H\).
\end{defi}
In the remainder of this article we will consider only first-order reductive Cartan geometries. Observe that we have \(\mathfrak{z} \cong \mathfrak{g}/\mathfrak{h}\) as representations of \(H\), so that we can replace all occurrences of \(\mathfrak{g}/\mathfrak{h}\) in the construction of the fake tangent and fake frame bundles above by \(\mathfrak{z}\).

Making use of the split of \(\mathfrak{g}\), we see that the Cartan connection \(A \in \Omega^1(\mathcal{P}, \mathfrak{g})\) decomposes in the form \(A = \omega + e\), where \(\omega \in \Omega^1(\mathcal{P}, \mathfrak{h})\) and \(e \in \Omega^1(\mathcal{P}, \mathfrak{z})\). From the latter we can obtain a coframe field \(\varepsilon: TM \to \mathcal{T}\) by the following construction. Let \(x \in M\) and \(v \in T_xM\). Choose \(p \in \mathcal{P}_x\) and \(w \in T_p\mathcal{P}\) such that \(\pi_*(w) = v\). Finally, define
\begin{equation}\label{eqn:tmfake}
\varepsilon(v) = p(e(w))\,.
\end{equation}
To see that this definition is independent of the choice of \(w \in T_p\mathcal{P}\), let \(w' \in T_p\mathcal{P}\) with \(\pi_*(w') = v\). The difference \(w - w'\) then lies in the kernel of \(\pi_*\), and is thus of the form \(\tilde{h}(p)\) for some \(h \in \mathfrak{h}\). According to property~\ref{cartan:mcform} of a Cartan connection, we have \(A(\tilde{h}) = h \in \mathfrak{h}\), and thus, \(e(w - w') = 0\). Further, the definition is also independent of the choice of \(p \in \mathcal{P}_x\), which follows from the equivariance property~\ref{cartan:equivariant}. Thus, the map \(\varepsilon\) is well-defined. One easily checks that \(\varepsilon\) is indeed a vector bundle isomorphism as a consequence of condition~\ref{cartan:isomorphism}.

The coframe field now allows relating the fake frame bundle \(\mathcal{P}\) and the frame bundle \(\mathrm{GL}(M)\). Recall that each element \(p \in \mathcal{P}_x\) can be interpreted as an admissible frame, i.e., a map \(p: \mathfrak{z} \to \mathcal{T}_x\). Compositing this map with the inverse of the coframe field we then obtain a map \(\varepsilon^{-1} \circ p: \mathfrak{z} \to T_xM\), which we call an admissible frame of \(T_xM\). This notion is compatible with our previous definition~\eqref{eqn:glframebundle} of a frame as a linear bijection \(f: \mathbb{R}^n \to T_xM\), since \(\varepsilon^{-1} \circ p\) is a linear bijection and \(\mathfrak{z} \cong \mathbb{R}^n\) as vector spaces. The assignment of an admissible frame to each \(p \in \mathcal{P}\) is a map \(\phi: \mathcal{P} \to \mathrm{GL}(M)\), which is injective, but in general not surjective. Its image \(P \subset \mathrm{GL}(M)\) is a subbundle of \(\mathrm{GL}(M)\) which is isomorphic to \(\mathcal{P}\), and hence also a principal $H$-bundle with the group action given by \(R_hf = f \circ \mathrm{Ad}(h)\). We call \(P\) the \emph{admissible frame bundle} of the Cartan geometry \((\pi: \mathcal{P} \to M, A)\). Making once more use of the map \(\phi\), it is equipped with the structure of a Cartan geometry \((\tilde{\pi}: P \to M, \tilde{A})\), where \(\tilde{\pi} = \pi \circ \phi^{-1}\) is the restriction of the bundle map of \(\mathrm{GL}(M)\) to \(P\) and \(\tilde{A} = (\phi^{-1})^*(A)\).

We conclude our discussion of first-order reductive Cartan geometries with the remark that the $\mathfrak{z}$-valued part \(\tilde{e} = (\phi^{-1})^*e\) of the Cartan connection on \(P\) has a particular form. Let \(p \in \mathcal{P}\) and \(w \in T_p\mathcal{P}\), as well as \(\tilde{p} = \phi(p)\) and \(\tilde{w} = \phi_*(w)\). By the definition of the pullback we have \(\tilde{e}(\tilde{w}) = e(w)\). Further, from the fact that \(\phi\) is a bundle morphism follows that \(\tilde{\pi}_*(\tilde{w}) = \pi_*(w)\). Finally, using the definition~\eqref{eqn:tmfake} of the map \(\varepsilon\), we find that
\begin{equation}\label{eqn:solderform}
\tilde{e}(\tilde{w}) = e(w) = p^{-1}(\varepsilon(\pi_*(w))) = (\varepsilon^{-1} \circ p)^{-1}(\tilde{\pi}_*(\tilde{w})) = \tilde{p}^{-1}(\tilde{\pi}_*(\tilde{w}))\,.
\end{equation}
The 1-form \(\tilde{e}\) defined by the expression on the right hand side is called the \emph{solder form}, or canonical form, which we denoted \(\theta\) at the end of section~\ref{subsec:lifts}.

\subsection{Definition of symmetries in Cartan language}\label{subsec:cartansym}
We finally introduce a notion of invariance of a first-order reductive Cartan geometry \((\pi: \mathcal{P} \to M, A)\) under diffeomorphisms of the underlying manifold \(M\). Recall from section~\ref{subsec:lifts} that any diffeomorphism \(\varphi: M \to M\) induces a diffeomorphism \(\bar{\varphi}: \mathrm{GL}(M) \to \mathrm{GL}(M)\), and that this relation yields the notion of the complete lift \(\bar{\xi}\) of a vector field \(\xi\) from \(M\) to \(\mathrm{GL}(M)\), where \(\xi\) and \(\bar{\xi}\) are the generators of diffeomorphisms on \(M\) and \(\mathrm{GL}(M)\), respectively. We define:
\begin{defi}[Invariant Cartan geometry]\label{def:invcartgeom}
A first-order reductive Cartan geometry \((\pi: \mathcal{P} \to M, A)\) is called \emph{invariant} under the group of diffeomorphisms generated by a vector field \(\xi\) on \(M\) if \(\bar{\xi}\) is tangent to the admissible frame bundle \(P \subset \mathrm{GL}(M)\) and the Cartan connection \(\tilde{A}\) is invariant under the flow of \(\bar{\xi}\) restricted to \(P\), \(\mathcal{L}_{\bar{\xi}}\tilde{A} = 0\), where \(P\) and \(\tilde{A}\) are defined as in section~\ref{subsec:forcartan}.
\end{defi}
Note that we could equally well have used the associated vector fields \(\underline{\tilde{A}} = \phi_*\underline{A}\) in order to define the invariance of a Cartan geometry. Depending on the particular example it is more convenient to work with one or the other definition.

As one can see from the definition above, the invariance of the Cartan geometry \((\pi: \mathcal{P} \to M, A)\) under a vector field \(\xi\) on \(M\) is completely defined in terms of the Cartan geometry \((\tilde{\pi}: P \to M, \tilde{A})\), which we derived from the original Cartan geometry using the procedure detailed in the preceding section. Further, if the original Cartan geometry is already defined on a subbundle \(\mathcal{P} = P\) of the frame bundle, and the Cartan connection \(A\) is chosen such that \(e\) is the solder form, we find that \(\varepsilon\) and \(\phi\) are the identity maps on their respective domains and \(\tilde{A} = A\). For this reason, we will simplify our notation and drop the tilde in the rest of this article, i.e., we will assume that we start from a Cartan geometry \((\pi: P \to M, A)\) defined on \(P \subset \mathrm{GL}(M)\) with \(e\) being the solder form on \(P\).

Note that the solder form \(e\) is defined on the whole frame bundle as \(e(w) = p^{-1}(\pi_*(w))\) for \(p \in \mathrm{GL}(M)\) and \(w \in T_p\mathrm{GL}(M)\), and that it is invariant under the frame bundle lift \(\bar{\varphi}\) of any diffeomorphism \(\varphi: M \to M\). To see this, recall from section~\ref{subsec:lifts} that \(\bar{\varphi}\) acts on \(\mathrm{GL}(M)\) via composition with \(\varphi_*\). We thus have
\begin{multline}
\bar{\varphi}^*(e)(w) = e(\bar{\varphi}_*(w)) = \bar{\varphi}(p)^{-1}(\pi_*(\bar{\varphi}_*(w))) = \left[p^{-1} \circ (\varphi_*)^{-1} \circ \pi_* \circ \bar{\varphi}_*\right](w)\\
= \left[p^{-1} \circ (\varphi^{-1} \circ \pi \circ \bar{\varphi})_*\right](w) = p^{-1}(\pi_*(w)) = e(w)\,,
\end{multline}
and therefore \(\bar{\varphi}^*(e) = e\). As a consequence, the Lie derivative \(\mathcal{L}_{\bar{\xi}}e\) vanishes for the complete lift \(\bar{\xi}\) of any vector field \(\xi\) on \(M\). In order to check the invariance of a Cartan connection as stated in the definition above it is thus sufficient to check that \(\bar{\xi}\) is tangent to \(P\) and that \(\mathcal{L}_{\bar{\xi}}\omega = 0\).

We finally remark that there is a close relation between the invariance of Cartan geometries which we use here and morphisms of Cartan geometries as given in section~\ref{subsec:cartan}. If a first-order reductive Cartan geometry \((\pi: P \to M, A)\) is invariant under a vector field \(\xi\) on \(M\), then the restriction of the frame bundle lift \(\bar{\xi}\) generates a one-parameter group of automorphisms of Cartan geometries on \(P\). However, the converse is not true: not every one-parameter group of automorphisms of Cartan geometries is generated by the frame bundle lift of a vector field on \(M\). In this work we focus on those symmetries which are generated and completely defined by a vector field on \(M\), and can thus be interpreted as symmetries of spacetime.

The definition of invariance of first-order reductive Cartan geometries provides us with the general framework we will be using throughout the remainder of this article. In the following we will apply this framework to particular examples of Cartan geometries, which correspond to other structures conventionally used to describe the geometry of spacetime. The first example is affine geometry, which we discuss in the following section.

\section{Affine Cartan geometries}\label{sec:affine}
After having discussed the notion of symmetry for a general first-order reductive Cartan geometry, we now come to the discussion of Cartan geometries based on particular model Klein geometries \(G/H\). As we have seen in the previous section, any first-order reductive Cartan geometry gives rise to a Cartan geometry on a subbundle \(P\) of the general linear frame bundle \(\mathrm{GL}(M)\) of the underlying manifold \(M\). The largest possible subbundle \(P\) is of course the frame bundle \(\mathrm{GL}(M)\) itself. This is the case we study in this section, where \(G\) will be an affine group and \(H\) will be the corresponding general linear group. We start with a brief overview of this model Klein geometry and its relation to affine connections on \(M\) in section~\ref{subsec:affklein}. We then show in section~\ref{subsec:affconnsym} how the notion of symmetry we discussed in the preceding section relates to the symmetry of an affine connection as given in section~\ref{subsec:lie}. The purpose of this section is rather illustrative, as this is the most simple model geometry we consider in this article. We will discuss more advanced models in the following sections.

\subsection{The affine model geometry}\label{subsec:affklein}
In this section we discuss Cartan geometries based on the model \(G/H\), where \(G = \mathrm{Aff}(n, \mathbb{R}) = \mathrm{GL}(n, \mathbb{R}) \ltimes \mathbb{R}^n\) is the general affine group and \(H = \mathrm{GL}(n, \mathbb{R})\) is the general linear group. We therefore briefly review the Lie algebra structure of \(\mathfrak{aff}(n, \mathbb{R})\) and \(\mathfrak{gl}(n, \mathbb{R})\). Note that this model is reductive, as required for the construction in the previous section, so that \(\mathfrak{aff}(n, \mathbb{R})\) splits into a direct sum \(\mathfrak{aff}(n, \mathbb{R}) = \mathfrak{gl}(n, \mathbb{R}) \oplus \mathbb{R}^n\). We can thus uniquely write every element \(a \in \mathfrak{aff}(n, \mathbb{R})\) in the form
\begin{equation}\label{eqn:affalgbasis}
a = h + z = h^i{}_j\mathcal{H}_i{}^j + z^i\mathcal{Z}_i\,,
\end{equation}
where \(\mathcal{H}_i{}^j\) is the matrix with entries 1 in the $i$th row and $j$th column and 0 otherwise, and \(\mathcal{Z}_i\) is the canonical basis of \(\mathbb{R}^n\). The Lie bracket of \(\mathfrak{aff}(n, \mathbb{R})\) then follows from the Lie brackets of the basis elements, which take the form
\begin{equation}\label{eqn:affalgebra}
[\mathcal{H}_i{}^j,\mathcal{H}_k{}^l] = \delta^j_k\mathcal{H}_i{}^l - \delta^l_i\mathcal{H}_k{}^j\,,\quad
[\mathcal{H}_i{}^j,\mathcal{Z}_k] = \delta^j_k\mathcal{Z}_i\,,\quad
[\mathcal{Z}_i,\mathcal{Z}_j] = 0\,.
\end{equation}
We can use the same basis~\eqref{eqn:affalgbasis} to expand the Cartan connection \(A\) in the form
\begin{equation}\label{eqn:affcartconn}
A = \omega + e = \omega^i{}_j\mathcal{H}_i{}^j + e^i\mathcal{Z}_i\,,
\end{equation}
using the notation introduced in the previous section. Similarly, we can expand the associated vector fields \(\underline{A}\) in the dual basis, which is reflected by the notation
\begin{equation}\label{eqn:afffundvect}
\underline{A}\left(h^i{}_j\mathcal{H}_i{}^j + z^i\mathcal{Z}_i\right) = h^i{}_j\underline{\omega}_i{}^j + z^i\underline{e}_i\,.
\end{equation}
We will use these basis expansions throughout the remainder of this section.

We now turn our focus from the structure of the affine Lie algebra to the geometry of the frame bundle \(P = \mathrm{GL}(M)\). In particular, we are interested in its vertical tangent spaces and the one-forms \(\mu_x \in \Omega^1(P_x, \mathfrak{h})\) induced by the Maurer-Cartan form \(\mu \in \Omega^1(H, \mathfrak{h})\), which are relevant for the definition of a Cartan geometry, as stated in section~\ref{subsec:cartan}. Using the notations above and the induced coordinates \((x^a,f_i^a)\) on \(P\) introduced in section~\ref{subsec:lifts} the induced one-forms \(\mu_x\) and the fundamental vector fields \(\widetilde{\mathcal{H}}_i{}^j\) are given by~\cite{Kobayashi:1963}
\begin{equation}\label{eqn:affmccan}
\mu_x = f^{-1}{}^i_adf_j^a\mathcal{H}_i{}^j \quad \text{and} \quad \widetilde{\mathcal{H}}_i{}^j = f_i^a\bar{\partial}^j_a\,.
\end{equation}
Using these expressions, is helpful to express the conditions~\ref{cartan:equivariant} and~\ref{cartan:mcform} on the Cartan connection \(A\), as well as the conditions~\ref{cartan:equivariant2} and~\ref{cartan:canonical} on the associated vector fields \(\underline{A}\) in coordinates for later use. The conditions~\ref{cartan:mcform} and~\ref{cartan:canonical} then take the form
\begin{equation}\label{eqn:affmccan2}
e^i(\bar{\partial}^j_a) = 0\,, \quad \omega^i{}_j(\bar{\partial}^k_a) = \delta^k_jf^{-1}{}^i_a \qquad \text{and} \qquad \underline{\omega}_i{}^j = f_i^a\bar{\partial}^j_a\,.
\end{equation}
The conditions~\ref{cartan:equivariant} and~\ref{cartan:equivariant2} can most easily be expressed in differential form. Condition~\ref{cartan:equivariant} then takes the form
\begin{equation}
\mathcal{L}_{\widetilde{\mathcal{H}}_i{}^j}\omega = -\omega^j{}_k\mathcal{H}_i{}^k + \omega^k{}_i\mathcal{H}_k{}^j\,, \quad
\mathcal{L}_{\widetilde{\mathcal{H}}_i{}^j}e = -e^j\mathcal{Z}_i\,,\label{eqn:affeqv}
\end{equation}
while condition~\ref{cartan:equivariant2} reads
\begin{equation}
\mathcal{L}_{\widetilde{\mathcal{H}}_i{}^j}\underline{\omega}_k{}^l = \delta^j_k\underline{\omega}_i{}^l - \delta^l_i\underline{\omega}_k{}^j\,, \quad \mathcal{L}_{\widetilde{\mathcal{H}}_i{}^j}\underline{e}_k = \delta^j_k\underline{e}_i\,.\label{eqn:affeqv2}
\end{equation}
Of course this form is sufficient only if \(H\) is connected, which is not the case for \(H = \mathrm{GL}(n, \mathbb{R})\). It is possible to circumvent this problem by considering only the oriented frame bundle \(\mathrm{GL}^+(M)\) equipped with a Cartan geometry modeled on \(G^+/H^+ = \mathrm{Aff}^+(n, \mathbb{R})/\mathrm{GL}^+(n, \mathbb{R})\), but this is possible only if \(M\) is orientable. In the more general, non-orientable case one may simply demand that in addition to the differential form of conditions~\ref{cartan:equivariant} and~\ref{cartan:equivariant2} as given above also their original form holds for one (and hence for any) reflection.

\subsection{Symmetries of affine connections}\label{subsec:affconnsym}
After discussing general properties of the affine model geometry we now briefly discuss the relation between Cartan geometries modeled on this Klein geometry and affine connections on \(M\), whose role in gravitational physics dates back to early attempts to unify gravity and electrodynamics~\cite{Einstein:1923a,Einstein:1923b,Eddington:1924}. Note that there are different ways to describe an affine connection. The most widely used description in gravitational physics is given in terms of a Koszul connection, or covariant derivative \(\nabla\) acting on tensor fields, or more generally speaking on sections of a vector bundle which is associated to the frame bundle~\cite{Carroll:2004}. In the context of gauge theories, however, a description in terms of a principal \(\mathrm{GL}(n, \mathbb{R})\) connection form \(\omega \in \Omega^1(P,\mathfrak{h})\) on the frame bundle \(P = \mathrm{GL}(M)\) is more common~\cite{Giachetta:2009}. For obvious reasons, we can immediately identify this principal connection with the $\mathfrak{h}$-valued part \(\omega\) of the Cartan connection on \(P\), since the defining properties of a principal connection are exactly conditions~\ref{cartan:equivariant} and~\ref{cartan:mcform}, restricted to \(\omega\). Conversely, any principal \(\mathrm{GL}(n, \mathbb{R})\) connection \(\omega\) on the frame bundle together with the solder form \(e\) given in equation~\eqref{eqn:solderform} yields a Cartan connection \(A = \omega + e\) on \(P\). One easily checks that \(A\) satisfies all conditions on a Cartan geometry.

Since in section~\ref{subsec:lie} we have introduced the notion of symmetry for a connection in terms of a covariant derivative \(\nabla\) with connection coefficients \(\Gamma^a{}_{bc}\), we briefly review the relation of this description to the formulation as a principal bundle connection \(\omega\), starting from the covariant derivative side. Let \(\gamma: \mathbb{R} \to P\) be a curve on \(P\) and \(z \in \mathfrak{z}\). Then \(\gamma_z: \mathbb{R} \to TM, t \mapsto \gamma(t)(z)\) is a curve on \(TM\), which we can also regard a vector field defined on the projected curve \(\pi \circ \gamma\) on \(M\). We can take the covariant derivative of this vector field along the curve and use the inverse frames \((\gamma(t))^{-1}\) to map the result of this operation back to \(\mathfrak{z}\). Finally, we define \(\omega(\gamma'(t))\) as the unique element of \(\mathfrak{h}\) such that
\begin{equation}
\mathrm{ad}(\omega(\gamma'(t)))(z) = (\gamma(t))^{-1}\left(\nabla_{(\pi \circ \gamma)'(t)}\gamma(t)(z)\right)
\end{equation}
for all \(z \in \mathfrak{z}\). Using the basis expansion~\eqref{eqn:affalgbasis} and the induced coordinates on the frame bundle introduced in section~\ref{subsec:lifts} the coordinate expression of the Cartan connection then takes the simple form
\begin{equation}\label{eqn:affcartconn2}
e^i = f^{-1}{}^i_adx^a\,, \quad \omega^i{}_j = f^{-1}{}^i_adf_j^a + f^{-1}{}^i_af_j^b\Gamma^a{}_{bc}dx^c\,,
\end{equation}
which are the well-known expressions for the solder form and the connection form of an affine connection on the frame bundle~\cite{Kobayashi:1963}. The associated vector fields are given by
\begin{equation}\label{eqn:afffundvect2}
\underline{e}_i = f_i^a\left(\partial_a - f_j^b\Gamma^c{}_{ba}\bar{\partial}^j_c\right) \,, \quad \underline{\omega}_i{}^j = f_i^a\bar{\partial}^j_a\,,
\end{equation}
where the former are the standard horizontal vector fields, and the latter the fundamental vector fields on the frame bundle~\cite{Kobayashi:1963}.

The opposite direction of this construction is done as follows. Consider the exterior covariant derivative \(d^{\omega}\), whose action on a general $H$-equivariant horizontal $k$-form \(\alpha\) on \(P\) with values in a representation vector space \(V\) with representation \(\rho: H \to \mathrm{GL}(V)\) is given by
\begin{equation}
d^{\omega}\alpha = d\alpha + \rho_*(\omega) \wedge \alpha\,,
\end{equation}
and yields a $H$-equivariant $V$-valued horizontal $(k + 1)$-form \(d^{\omega}\alpha\). Here, the term \emph{horizontal} means that these forms vanish on vertical vectors. Note that \(\rho_*: \mathfrak{h} \to \mathfrak{gl}(V)\), so that \(\rho_*(\omega)\) is a one-form on \(P\) with values in \(\mathfrak{gl}(V)\), which can thus applied to a $V$-valued $k$-form \(\alpha\). To relate this exterior covariant derivative to a covariant derivative on \(M\), let \(X\) be a vector field on \(M\) and define \(\bar{X}: P \to \mathfrak{z}\) as \(\bar{X}(p) = p^{-1}(X(\pi(p)))\). This is a $H$-equivariant $\mathfrak{z}$-valued (and vacuously horizontal) zero-form on \(P\), so that we can apply the exterior covariant derivative \(d^{\omega}\). Further, let \(Y\) be another vector field on \(M\) and \(Y^H\) its horizontal lift to \(P\), i.e., the unique vector field on \(P\) such that \(Y \circ \pi = \pi_* \circ Y^H\) and \(\omega(Y^H) = 0\). Then \(\bar{Z} = (d^{\omega}X)(Y^H)\) with \(\bar{Z}(p) = p^{-1}(Z(\pi(p)))\) for another vector field \(Z\) on \(M\). This vector field is the covariant derivative \(Z = \nabla_YX\). The connection coefficients follow by taking the coordinate vector fields \(\partial_a\) in place of \(X\) and \(Y\). A simple calculation shows that
\begin{equation}\label{eqn:cartconnaff}
\Gamma^a{}_{bc} = dx^a(\nabla_{\partial_c}\partial_b) = f_i^af^{-1}{}^j_b\omega^i{}_j(\partial_c)\,.
\end{equation}
Note that the right hand side of this equation apparently depends on the fiber coordinates \(f_i^a\) on \(P\), while the left hand side depends only on the manifold coordinates \(x^a\) on \(M\). However, the explicit dependence on \(f_i^a\) on the right hand side are exactly compensated by the equivariance (and thus implicit dependence on \(f_i^a\)) of \(\omega^i{}_j\), so that the right hand side is constant along the fibers of \(P\) and thus defines a function on the manifold \(M\).

With these relations we now come to the central proposition of this section, which we formulate as follows.
\begin{prop}\label{pro:affinv}
A Cartan geometry \((\pi: P \to M, A)\) modeled on the affine Klein geometry is invariant under a vector field \(\xi\) on \(M\) if and only if the corresponding affine connection \(\nabla\) is invariant under \(\xi\), \(\mathcal{L}_{\xi}\Gamma = 0\).
\end{prop}
\begin{proof}
First, observe that \(P = \mathrm{GL}(M)\), so that the complete lift \(\bar{\xi}\) of any vector field \(\xi\) to \(\mathrm{GL}(M)\) is trivially tangent to \(P\). Further, the solder form \(e\) is invariant under \(\bar{\xi}\), as shown in section~\ref{subsec:cartansym}. We therefore only need to prove that \(\omega\) is invariant under \(\bar{\xi}\) whenever \(\nabla\) is invariant under \(\xi\) and vice versa. We can do so by direct calculation of the Lie derivative. From equation~\eqref{eqn:affcartconn2} we find
\begin{equation}
\mathcal{L}_{\bar{\xi}}\omega^i{}_j = \mathcal{L}_{\bar{\xi}}(f^{-1}{}^i_adf_j^a + f^{-1}{}^i_af_j^b\Gamma^a{}_{bc}dx^c) = f^{-1}{}^i_af_j^b(\mathcal{L}_{\xi}\Gamma)^a{}_{bc}dx^c\,,
\end{equation}
which proves the first part of the proposition. The proof of the second part of the proposition follows from equation~\eqref{eqn:cartconnaff}. Its Lie derivative is given by
\begin{equation}
(\mathcal{L}_{\xi}\Gamma)^a{}_{bc} = f_i^af^{-1}{}^j_b(\mathcal{L}_{\bar{\xi}}\omega^i{}_j)(\partial_c)\,,
\end{equation}
where the relations~\eqref{eqn:affmccan2} and~\eqref{eqn:affeqv} enter the calculation of the right hand side.
\end{proof}

This proof concludes our rather illustrative discussion of affine Cartan geometries. Making use of the constructions detailed in this section we can now discuss the symmetries of other model geometries more often encountered in gravitational physics, and of orthogonal Cartan geometries in particular. This will be done in the following section.

\section{Orthogonal Cartan geometries}\label{sec:orthogonal}
After discussing the general construction relating spacetime symmetries of an affine connection and its Cartan geometry we now come to the special case of metric compatible connections, which is the most common case in theories of gravity. These can most conveniently be described by Cartan geometries modeled on the orthogonal Klein geometry, which we briefly review in section~\ref{subsec:ortklein}. We then discuss different geometries based on this model. The most general geometry of this type is Riemann-Cartan geometry, which we discuss in section~\ref{subsec:rcsym}. Two more specialized examples will be Riemannian geometry discussed in section~\ref{subsec:riemsym} and Weizenb\"ock geometry discussed in section~\ref{subsec:weizsym}.

\subsection{The orthogonal model geometries}\label{subsec:ortklein}
The Cartan geometries we discuss in this section can be modeled on any of the following three Klein geometries \(G/H\), where \(G\) and \(H\) are the groups
\begin{equation}\label{eqn:ortgroups}
G = \begin{cases}
\mathrm{O}(m,n+1) & \text{for } \Lambda = 1\\
\mathrm{IO}(m,n) & \text{for } \Lambda = 0\\
\mathrm{O}(m+1,n) & \text{for } \Lambda = -1
\end{cases}\,, \quad H = \mathrm{O}(m,n)\,,
\end{equation}
where \(\mathrm{IO}(m,n) = \mathrm{O}(m,n) \ltimes \mathbb{R}^{m,n}\) is the generalized Poincar\'e group. Here we introduced a parameter \(\Lambda \in \{-1, 0, 1\}\), which indicates the scalar curvature of the corresponding homogeneous space \(G/H\), and can be interpreted as (the sign of) the cosmological constant~\cite{Gielen:2012fz}. This interpretation becomes clear for the case \(m = 1, n = 3\) conventionally encountered in gravitational physics. In this case the homogeneous spaces of the models listed above are de Sitter space, Minkowski space and anti de Sitter space, respectively. Each of these spaces is equipped with a metric \(\eta\) of signature \((m,n)\), whose sign we choose so that we have \(m\) times ``$-$'' and \(n\) times ``$+$''. All models are reductive, so that the Lie algebras \(\mathfrak{g}\) split in the form \(\mathfrak{g} = \mathfrak{h} \oplus \mathfrak{z}\) into subrepresentations of the adjoint representation of \(H\) on \(\mathfrak{g}\). This split allows us to expand \(a \in \mathfrak{g}\) in the form
\begin{equation}\label{eqn:ortalgbasis}
a = h + z = \frac{1}{2}h^i{}_j\mathcal{H}_i{}^j + z^i\mathcal{Z}_i\,,
\end{equation}
where \(\mathcal{H}_i{}^j\) are the generators of \(\mathfrak{h}\) and \(\mathcal{Z}_i\) are the generators of translations, which span \(\mathfrak{z}\). We introduced a factor \(\frac{1}{2}\) here for convenience. Note that \(\mathfrak{h}\) is always a subalgebra of \(\mathfrak{g}\), while \(\mathfrak{z}\) is a subalgebra only for \(\Lambda = 0\). This can also be read off from the Lie algebra relations, which take the form
\begin{gather}
[\mathcal{H}_i{}^j,\mathcal{H}_k{}^l] = \delta^j_k\mathcal{H}_i{}^l - \delta^l_i\mathcal{H}_k{}^j + \eta_{ik}\eta^{lm}\mathcal{H}_m{}^j - \eta^{jl}\eta_{km}\mathcal{H}_i{}^m\,,\label{eqn:ortalgebra}\\
[\mathcal{H}_i{}^j,\mathcal{Z}_k] = \delta^j_k\mathcal{Z}_i - \eta_{ik}\eta^{jl}\mathcal{Z}_l\,, \qquad [\mathcal{Z}_i,\mathcal{Z}_j] = \Lambda\eta_{ik}\mathcal{H}_j{}^k\,,\nonumber
\end{gather}
making use of the parameter \(\Lambda\).

Before we discuss Cartan geometries based on this model Klein geometry, we make a few general statements on the structure of the admissible frame bundle \(P \subset \mathrm{GL}(M)\). Note that \(\mathfrak{z}\) is naturally equipped with a metric \(\eta(\mathcal{Z}_i, \mathcal{Z}_j) = \eta_{ij}\), which is invariant under the adjoint representation of \(H\). Since the group action of \(h \in H\) on a frame \(p \in P_x\), which is a linear bijection \(p: \mathfrak{z} \to T_xM\) with \(x = \pi(p) \in M\), is given by the adjoint representation \(R_hp = p \circ \mathrm{Ad}(h)\) as a consequence of the construction detailed in section~\ref{subsec:forcartan}, it follows that \(T_xM\) is equipped with a non-degenerate metric \(g(.,.) = \eta(p^{-1}(.), p^{-1}(.))\) of signature \((m,n)\), which is invariant under the action of \(H\) on \(P\) and thus independent of the choice of \(p \in P_x\). This turns \(M\) into a metric manifold. It follows further that \(P\) agrees with the orthonormal frame bundle, which can be written as
\begin{equation}\label{eqn:ortframebundle}
P = \mathrm{O}(M, g) = \bigcup_{x \in M}\{\text{linear isometries } p: \mathfrak{z} \to T_xM\} \subset \mathrm{GL}(M)\,.
\end{equation}
Using the coordinates \((x,f)\) introduced in section~\ref{subsec:lifts} it follows that
\begin{equation}\label{eqn:ortframebundle2}
(x,f) \in P \quad \Leftrightarrow \quad f_i^af_j^bg_{ab}(x) = \eta_{ij}\,.
\end{equation}
We will use the coordinates \((x^a,f_i^a)\) also to denote elements of \(P\), taking into account that only those frames are admissible which satisfy the condition given above. We finally remark that in the case \(\Lambda = 0\) corresponding to the group \(G = \mathrm{IO}(m,n)\) the algebra relations are invariant under a rescaling \(\mathcal{Z}_i \mapsto \lambda\mathcal{Z}_i\), so that the metric \(\eta\) and thus also \(g\) are determined only up to a constant factor.

We now come to the description of Cartan geometries on the orthonormal frame bundle. Using the component expansion introduced above, we write the Cartan connection as
\begin{equation}\label{eqn:ortcartconn}
A = \omega + e = \frac{1}{2}\omega^i{}_j\mathcal{H}_i{}^j + e^i\mathcal{Z}_i\,,
\end{equation}
while the associated vector fields take the form
\begin{equation}\label{eqn:ortfundvect}
\underline{A}\left(\frac{1}{2}h^i{}_j\mathcal{H}_i{}^j + z^i\mathcal{Z}_i\right) = \frac{1}{2}h^i{}_j\underline{\omega}_i{}^j + z^i\underline{e}_i\,.
\end{equation}
Note that we used the same factor \(\frac{1}{2}\) here which we introduced in the expansion~\eqref{eqn:ortalgbasis} and which we did not use in the corresponding definitions~\eqref{eqn:affcartconn} and~\eqref{eqn:afffundvect} for the affine model. The one-forms \(\mu_x \in \Omega^1(P_x, \mathfrak{h})\) induced by the Maurer-Cartan form \(\mu \in \Omega^1(H, \mathfrak{h})\) and the fundamental vector fields \(\widetilde{\mathcal{H}}_i{}^j\) of the generators \(\mathcal{H}_i{}^j\) are given by
\begin{equation}\label{eqn:ortmccan}
\mu_x = \frac{1}{2}f^{-1}{}^i_adf_j^a\mathcal{H}_i{}^j \quad \text{and} \quad \widetilde{\mathcal{H}}_i{}^j = f_i^a\bar{\partial}^j_a - \eta_{ik}\eta^{jl}f_l^a\bar{\partial}^k_a\,.
\end{equation}
Making use of these expressions, we can write the conditions~\ref{cartan:mcform} and~\ref{cartan:canonical} in the form
\begin{equation}\label{eqn:ortmccan2}
e^i(\bar{\partial}^j_a - \eta_{mk}\eta^{jl}f^{-1}{}^m_af_l^b\bar{\partial}^k_b) = 0\,, \quad \omega^i{}_j(\bar{\partial}^k_a - \eta_{nl}\eta^{km}f^{-1}{}^n_af_m^b\bar{\partial}^l_b) = \delta^k_jf^{-1}{}^i_a - \eta_{jl}\eta^{ik}f^{-1}{}^l_a
\end{equation}
and
\begin{equation}\label{eqn:ortmccan3}
\underline{\omega}_i{}^j = f_i^a\bar{\partial}^j_a - \eta_{ik}\eta^{jl}f_l^a\bar{\partial}^k_a\,.
\end{equation}
Here we have written the vertical vector fields on \(P\) in the form \(\bar{\partial}^j_a - \eta_{mk}\eta^{jl}f^{-1}{}^m_af_l^b\bar{\partial}^k_b\). Note that we cannot simply use the coordinate vector fields \(\bar{\partial}^j_a\), as these are not tangent to \(P\). We finally express the conditions~\ref{cartan:equivariant} and~\ref{cartan:equivariant2} in differential form. Condition~\ref{cartan:equivariant} then reads
\begin{gather}
\mathcal{L}_{\widetilde{\mathcal{H}}_i{}^j}\omega = -\frac{1}{2}\left(\omega^j{}_k\mathcal{H}_i{}^k - \omega^k{}_i\mathcal{H}_k{}^j + \omega^k{}_l\eta_{ik}\eta^{lm}\mathcal{H}_m{}^j - \omega^k{}_l\eta^{jl}\eta_{km}\mathcal{H}_i{}^m\right)\,,\nonumber\\
\mathcal{L}_{\widetilde{\mathcal{H}}_i{}^j}e = -e^j\mathcal{Z}_i + e^k\eta_{ik}\eta^{jl}\mathcal{Z}_l\,,\label{eqn:orteqv}
\end{gather}
while condition~\ref{cartan:equivariant2} is given by
\begin{equation}
\mathcal{L}_{\widetilde{\mathcal{H}}_i{}^j}\underline{\omega}_k{}^l = \delta^j_k\underline{\omega}_i{}^l - \delta^l_i\underline{\omega}_k{}^j + \delta_{ik}\delta^{lm}\underline{\omega}_m{}^j - \delta^{jl}\delta_{km}\underline{\omega}_i{}^m\,, \quad \mathcal{L}_{\widetilde{\mathcal{H}}_i{}^j}\underline{e}_k = \delta^j_k\underline{e}_i - \delta_{ik}\delta^{jl}\underline{e}_l\,.\label{eqn:orteqv2}
\end{equation}
As we also remarked regarding the affine model geometry at the end of section~\ref{subsec:affklein}, the differential conditions given above are sufficient only if \(H\) is connected, and otherwise must be supplemented with the original conditions~\ref{cartan:equivariant} and~\ref{cartan:equivariant2} as given in section~\ref{subsec:cartan}.

\subsection{Symmetries for Riemann-Cartan spacetimes}\label{subsec:rcsym}
We now come to the description of particular spacetime geometries based on the model geometry detailed in the previous section, and to their symmetries formulated in the language of Cartan geometry. The first and most general example we discuss here is Riemann-Cartan spacetime, which provides the geometric background of Einstein-Cartan theory~\cite{Cartan:1922,Cartan:1923zea,Kibble:1961ba,Sciama:1964wt,Hehl:1976kj,Hehl:1994ue,Blagojevic:2002du,Blagojevic:2013xpa,Trautman:2006fp}, and which we define as follows. Let \(M\) be a manifold of dimension \(m + n\) equipped with a metric \(g\) of signature \((m,n)\) and a metric compatible connection \(\nabla\). The connection coefficients of \(\nabla\) then take the general form
\begin{equation}\label{eqn:rcgamma}
\Gamma^a{}_{bc} = \frac{1}{2}g^{ad}(\partial_bg_{dc} + \partial_cg_{bd} - \partial_dg_{bc} - g_{be}T^e{}_{cd} - g_{ce}T^e{}_{bd}) + \frac{1}{2}T^a{}_{cb}\,,
\end{equation}
where \(T^a{}_{bc}\) is the torsion tensor~\eqref{eqn:torsion}. Note that in contrast to Riemannian geometry, where the torsion-free Levi-Civita connection is chosen, the connection of Riemann-Cartan geometry has in general non-vanishing torsion.

The geometry of Riemann-Cartan spacetime can now easily be written in terms of a Cartan connection modeled on the orthogonal Klein geometry displayed in the preceding section. Let \(P = \mathrm{O}(M, g)\) denote the orthogonal frame bundle defined by the metric \(g\). The metric compatible covariant derivative with coefficients~\eqref{eqn:rcgamma} gives rise to a connection \(\omega \in \Omega^1(P, \mathfrak{h})\), where \(\mathfrak{h}\) is the Lie algebra of the Lorentz group. Together with the solder form \(e\) on \(P\), \(\omega\) combines into a Cartan connection. This construction is the same is displayed already in section~\ref{subsec:affconnsym}, so we do not repeat it here, and simply display the basis expressions for the Cartan connection
\begin{equation}\label{eqn:rccartconn}
e^i = f^{-1}{}^i_adx^a\,, \quad \omega^i{}_j = f^{-1}{}^i_adf_j^a + f^{-1}{}^i_af_j^b\Gamma^a{}_{bc}dx^c\,,
\end{equation}
and the associated vector fields
\begin{equation}\label{eqn:rcfundvect}
\underline{e}_i = f_i^a\left(\partial_a - f_j^b\Gamma^c{}_{ba}\bar{\partial}^j_c\right) \,, \quad \underline{\omega}_i{}^j = f_i^a\bar{\partial}^j_a - \eta_{ik}\eta^{jl}f_l^a\bar{\partial}^k_a\,.
\end{equation}
Conversely, every Cartan connection can be written in terms of a metric and a metric compatible connection. As we have argued in section~\ref{subsec:ortklein}, the admissible frame bundle \(P \subset \mathrm{GL}(M)\) of an orthogonal Cartan geometry determines a metric \(g\) on \(M\) such that \(P = \mathrm{O}(M, g)\). Finally, the $\mathfrak{h}$-valued part \(\omega\) of the Cartan connection is a principal connection on \(P\), and thus gives rise to a metric compatible connection on \(M\). The metric and the connection coefficients obtained from this construction take the form
\begin{equation}\label{eqn:cartconnrc}
g_{ab} = \eta_{ij}e^i(\partial_a)e^j(\partial_b)\,, \quad \Gamma^a{}_{bc} = dx^a(\nabla_{\partial_c}\partial_b) = f_i^af^{-1}{}^j_b\omega^i{}_j(\partial_c)\,.
\end{equation}
We make use of these relations to prove the following proposition.

\begin{prop}\label{pro:rcinv}
A Cartan geometry \((\pi: P \to M, A)\) modeled on the orthogonal Klein geometry is invariant under a vector field \(\xi\) on \(M\) if and only if the metric and torsion of the associated Riemann-Cartan geometry are invariant, \(\mathcal{L}_{\xi}g = 0\) and \(\mathcal{L}_{\xi}T = 0\).
\end{prop}
\begin{proof}
Recall from the definition in section~\ref{subsec:cartansym} that a first-order reductive Cartan geometry \((\pi: P \to M, A)\), where \(P\) is a subbundle of the frame bundle of \(M\), is invariant under \(\xi\) if and only if the complete lift \(\bar{\xi}\) of \(\xi\) to the frame bundle is tangent to \(P\) and preserves the Cartan connection, \(\mathcal{L}_{\bar{\xi}}A = 0\). In the case of orthogonal Cartan geometries discussed here we have already seen that there exists a one-to-one correspondence between non-degenerate metrics on \(M\) and choices of the admissible frame bundle \(P\). It follows from the relation~\eqref{eqn:ortframebundle2} that \(\bar{\xi}\) is tangent to \(P\) if and only if
\begin{equation}\label{eqn:orttanglift}
0 = \bar{\xi}(f_i^af_j^bg_{ab}) = f_i^af_j^b(\xi^c\partial_cg_{ab} + \partial_a\xi^cg_{cb} + \partial_b\xi^cg_{ac}) = f_i^af_j^b(\mathcal{L}_{\xi}g)_{ab}\,,
\end{equation}
and thus if and only if \(\mathcal{L}_{\xi}g = 0\).

We finally need to show that, given that equation~\eqref{eqn:orttanglift} is satisfied, the torsion is invariant under \(\xi\) if and only if \(A\) is invariant under \(\bar{\xi}\). A direct calculation of \(\mathcal{L}_{\bar{\xi}}\omega\) in terms of \(\mathcal{L}_{\xi}\Gamma\), in analogy to the proof of proposition~\ref{pro:affinv}, shows that \(A\) is invariant under \(\bar{\xi}\) if and only if \(\mathcal{L}_{\xi}\Gamma = 0\). Finally, from equation~\eqref{eqn:rcgamma} one derives
\begin{multline}
(\mathcal{L}_{\xi}\Gamma)^a{}_{bc} = -\frac{1}{2}g^{ae}g^{fd}(\mathcal{L}_{\xi}g)_{ef}(\partial_bg_{dc} + \partial_cg_{bd} - \partial_dg_{bc} - g_{be}T^e{}_{cd} - g_{ce}T^e{}_{bd}) + \frac{1}{2}(\mathcal{L}_{\xi}T)^a{}_{cb} + \frac{1}{2}g^{ad}\\
\cdot \left[\partial_b(\mathcal{L}_{\xi}g)_{dc} + \partial_c(\mathcal{L}_{\xi}g)_{bd} - \partial_d(\mathcal{L}_{\xi}g)_{bc} - (\mathcal{L}_{\xi}g)_{be}T^e{}_{cd} - (\mathcal{L}_{\xi}g)_{ce}T^e{}_{bd} - g_{be}(\mathcal{L}_{\xi}T)^e{}_{cd} - g_{ce}(\mathcal{L}_{\xi}T)^e{}_{bd}\right]\,,
\end{multline}
which vanishes if \(\mathcal{L}_{\xi}g = 0\) and \(\mathcal{L}_{\xi}T = 0\). Conversely, if \(\mathcal{L}_{\xi}\Gamma\) vanishes, it follows that also \(\mathcal{L}_{\xi}T\) vanishes due to equation~\eqref{eqn:torsion}.
\end{proof}

This concludes our discussion of Riemann-Cartan geometry. As we have seen, the notion of symmetry of a Riemann-Cartan geometry translates into the language of Cartan geometry in analogy to that of affine geometry discussed in the previous section.

\subsection{Symmetries for Riemannian spacetimes}\label{subsec:riemsym}
The most widely used subclass of Riemann-Cartan spacetime is Riemannian spacetime, which is obtained by fixing vanishing torsion, so that the coefficients~\eqref{eqn:rcgamma} of the connection \(\nabla\) reduce to the Levi-Civita connection
\begin{equation}\label{eqn:riemgamma}
\Gamma^a{}_{bc} = \frac{1}{2}g^{ad}(\partial_bg_{dc} + \partial_cg_{bd} - \partial_dg_{bc})\,.
\end{equation}
Despite its importance for gravitational physics, in particular for general relativity and related theories~\cite{Einstein:1915ca,Carroll:2004}, we only briefly discuss symmetries of Riemannian spacetimes in this article, since any statements we could make here follow immediately from the discussion of Riemann-Cartan spacetimes in the preceding section by setting \(T^a{}_{bc} \equiv 0\). For the Cartan geometry \((\pi: P \to M, A)\) this is equivalent to the condition
\begin{equation}
F = dA + \frac{1}{2}[A, A] \in \Omega^2(P, \mathfrak{h})\,,
\end{equation}
i.e., the $\mathfrak{z}$-valued part of the Cartan curvature \(F\) vanishes. Therefore, we only state the following proposition:

\begin{prop}\label{pro:rieminv}
A torsion-free Cartan geometry \((\pi: P \to M, A)\) modeled on the orthogonal Klein geometry is invariant under a vector field \(\xi\) on \(M\) if and only if the associated metric is invariant, \(\mathcal{L}_{\xi}g = 0\).
\end{prop}
\begin{proof}
This proposition follows immediately from proposition~\ref{pro:rcinv} for \(T = 0\).
\end{proof}

Thus, our formulation of symmetries in the language of Cartan geometry in particular includes the well-known notion of symmetries of a Riemannian spacetime in terms of Killing vector fields.

\subsection{Symmetries for Weizenb\"ock spacetimes}\label{subsec:weizsym}
The last example geometry based on the orthogonal model which we discuss here is Weizenb\"ock geometry. Although it has been used already by Einstein as the geometric background for his teleparallel gravity theory~\cite{Einstein:1928}, it has not received much attention until recently~\cite{Blagojevic:2002du,Blagojevic:2013xpa,Aldrovandi:2013wha,Baez:2012bn}. In its classical form it is based on the assumption that spacetime is a parallelizable manifold \(M\), so that it can be equipped with a global section of the frame bundle, the tetrad \(\theta\). The tetrad defines a notion of distant parallelism, or teleparallelism, by providing a way to compare tangent vectors to different points. The connection \(\nabla\) associated with this teleparallelism is the Weizenb\"ock connection with coefficients
\begin{equation}\label{eqn:weizgamma}
\Gamma^a{}_{bc} = \theta_i^a\partial_c\theta^i_b\,
\end{equation}
with respect to which the tetrad is constant. By construction it has vanishing curvature, but in general non-vanishing torsion. The tetrad further defines a metric,
\begin{equation}\label{eqn:weizmetric}
g_{ab} = \theta^i_a\theta^j_b\eta_{ij}\,,
\end{equation}
which is also constant with respect to the Weizenb\"ock connection.

It should now be clear how to construct a Cartan geometry modeled on the orthogonal Klein geometry from a given Weizenb\"ock geometry. The admissible frame bundle is simply the orthogonal frame bundle \(P = \mathrm{O}(M, g)\) defined by the metric \(g\), while the Cartan connection \(A\) is derived from the Weizenb\"ock connection following equations~\eqref{eqn:rccartconn}, and~\eqref{eqn:rcfundvect} for the associated vector fields. It should be noted that the Cartan geometry obtained from this construction does not change if we apply a global (i.e., independent of the position in spacetime) Lorentz transform \(\theta^i_a \mapsto L^i{}_j\theta^j_a\) with constant \(L \in \mathrm{O}(m,n)\) to the tetrad, since both the metric and the Weizenb\"ock connection are invariant under this transformation. Further, note that in this case the Cartan curvature satisfies
\begin{equation}
F = dA + \frac{1}{2}[A, A] \in \Omega^2(P, \mathfrak{z})\,,
\end{equation}
i.e., the $\mathfrak{h}$-valued component of \(F\), corresponding to the curvature of the Weizenb\"ock connection, vanishes. We remark that also more sophisticated constructions exist, which turn Weizenb\"ock geometry into a higher Cartan geometry, in the sense of 2-categories~\cite{Baez:2012bn}.

A less obvious construction is the derivation of a Weizenb\"ock geometry from a given orthogonal Cartan geometry \((\pi: P \to M, A)\). Using the same construction as given in section~\ref{subsec:rcsym} for a Riemann-Cartan geometry, one obtains the metric \(g\) and the coefficients of the connection \(\nabla\) using equation~\eqref{eqn:cartconnrc}. However, it is not possible to uniquely determine the tetrad from the Cartan geometry, but only up to a global Lorentz transform. To see this, choose a tetrad \(\theta(x)\) at some point \(x \in M\) which is orthonormal with respect to the metric \(g\), and use the connection \(\nabla\) to define \(\theta\) on all of \(M\) via parallel transport. Here we use the fact that the connection \(\nabla\) is flat, so that the parallel transport does not depend on the chosen path. Since \(\theta\) is covariantly constant by construction,
\begin{equation}
0 = \nabla_a\theta^i_b = \partial_a\theta^i_b - \Gamma^c{}_{ba}\theta^i_c\,,
\end{equation}
it follows immediately that \(\nabla\) is indeed the Weizenb\"ock connection with coefficients~\eqref{eqn:weizgamma}, so that we have constructed a Weizenb\"ock geometry.

From the fact that a global Lorentz transform of the tetrad \(\theta\) does not change the corresponding Cartan geometry one can already deduce that the Cartan geometry is invariant under diffeomorphisms which change the tetrad by a global Lorentz transform. For the purpose of this article we can thus state:

\begin{prop}\label{pro:weizinv}
The Cartan geometry \((\pi: P \to M, A)\) corresponding to a Weizenb\"ock geometry with tetrad \(\theta\) is invariant under a vector field \(\xi\) if and only if \(\xi\) acts on \(\theta\) as the generator of a global Lorentz transform, i.e., such that \((\mathcal{L}_{\xi}\theta)^i_a = \lambda^i{}_j\theta^j_a\) with constant \(\lambda \in \mathfrak{h} = \mathfrak{o}(m,n)\).
\end{prop}
\begin{proof}
A direct calculation shows that
\begin{equation}
(\mathcal{L}_{\xi}g)_{ab} = 2\eta_{ij}(\mathcal{L}_{\xi}\theta)^i_{(a}\theta^j_{b)}
\end{equation}
and
\begin{equation}
(\mathcal{L}_{\xi}\Gamma)^a{}_{bc} = \theta_i^a\partial_c(\mathcal{L}_{\xi}\theta)^i_b + \partial_c\theta^i_b(\mathcal{L}_{\xi}\theta)_i^a = \theta_i^a\partial_c(\mathcal{L}_{\xi}\theta)^i_b - \theta_i^a\theta_j^d\partial_c\theta^j_b(\mathcal{L}_{\xi}\theta)^i_d\,,
\end{equation}
which both vanish if \((\mathcal{L}_{\xi}\theta)^i_a = \lambda^i{}_j\theta^j_a\) with constant \(\lambda \in \mathfrak{o}(m,n)\). Using proposition~\ref{pro:rcinv} it follows further that also the corresponding Cartan geometry is invariant under \(\xi\).

The proof of the converse direction proceeds similarly. Once again making use of proposition~\ref{pro:rcinv} it follows from the invariance of a Cartan geometry \((\pi: P \to M, A)\) under a vector field \(\xi\) that also the metric \(g\) and connection \(\nabla\) of the corresponding Weizenb\"ock geometry are invariant under \(\xi\). Further, recall that we have constructed the tetrad \(\theta\) as the unique (up to a global Lorentz transform) tetrad which is covariantly constant with respect to \(\nabla\) and orthonormal with respect to \(g\). Let \(t \mapsto \varphi_t\) denote the one-parameter group of diffeomorphisms generated by \(\xi\) and \(\theta_i = \theta_i^a\partial_a\). We then have
\begin{equation}
0 = \nabla_{\theta_j}\theta_i = \varphi_{-t*} \circ \nabla_{\varphi_{t*} \circ \theta_j}(\varphi_{t*} \circ \theta_i)\,,
\end{equation}
since \(\nabla\) is invariant and \(\theta\) is covariantly constant. Using the facts that \(\varphi_{-t*}\) is a diffeomorphism and that the vector fields \(\varphi_{t*} \circ \theta_j\) span \(TM\) it follows that also \(\varphi_t^*\theta\) is covariantly constant. Further, \(\varphi_t^*\theta\) is orthonormal with respect to \(\varphi_t^*g = g\). However, since \(\theta\) is uniquely determined by these two conditions up to a global Lorentz transform, it follows that also \(\varphi_t^*\theta\) and \(\theta\) must be related by a global Lorentz transform, \((\varphi_t^*\theta)^i_a = L^i{}_j(t)\theta^j_a\). Differentiation with respect to \(t\) yields
\begin{equation}
(\mathcal{L}_{\xi}\theta)^i_a = \left.\frac{d}{dt}L^i{}_j(t)\right|_{t = 0}\theta^j_a = \lambda^i{}_j\theta^j_a
\end{equation}
with constant \(\lambda \in \mathfrak{o}(m,n)\).
\end{proof}

This result has an interesting implication: it shows that the notion of symmetry of Weizenb\"ock geometry, in the sense of proposition~\ref{pro:weizinv}, is (in general) not invariant under local Lorentz transforms of the form \(\theta^i_a(x) \mapsto \theta'^i_a(x) = L^i{}_j(x)\theta^j_a(x)\), although tetrads related by a local Lorentz transform are usually considered to be physically equivalent. In particular, if \(\theta\) is invariant under the diffeomorphisms generated by a vector field \(\xi\), so that \((\mathcal{L}_{\xi}\theta)^i_a = \lambda^i{}_j\theta^j_a\) with constant \(\lambda \in \mathfrak{o}(m,n)\), then \(\theta'\) is invariant if and only if
\begin{equation}
\left(\xi^a\partial_aL^i{}_k + L^i{}_l\lambda^l{}_k\right)L^{-1}{}^k{}_j = \lambda'^i{}_j\,,
\end{equation}
where \(\lambda' \in \mathfrak{o}(m,n)\) is also a constant infinitesimal Lorentz transform. This is due to the fact that only Lorentz transforms of this particular form leave the Weizenb\"ock connection invariant. This observation might be relevant in the context of teleparallel gravity and its possible extensions, where the construction of symmetric solutions is based on the choice of a suitable tetrad~\cite{Krssak:2015oua}.

This remark concludes our discussion of Cartan geometries based on the orthogonal model geometries, and of spacetime Cartan geometries in general. We have seen that we can express the symmetries of the most common geometric descriptions of spacetime in terms of Cartan geometry. In the remainder of this article we will see that Cartan geometry allows us to go further and describe the symmetries of observer space in the same way as those of spacetime.

\section{Observer space Cartan geometries}\label{sec:observer}
The geometries we discussed in the previous two sections have in common that the base manifold \(M\) was typically interpreted as spacetime. In this section we shift our perspective and consider Cartan geometries whose base manifold can rather be interpreted as the space of physical observers, and is hence called \emph{observer space}; see~\cite{Gielen:2012fz} for a detailed discussion. This physical interpretation, as well as the description of the geometry of observer space in terms of Cartan geometry, will be discussed in section~\ref{subsec:obscartan}. We will discuss symmetries of observer space in section~\ref{subsec:obssym}, and in particular answer the question which diffeomorphisms of observer space can be interpreted as diffeomorphisms of an underlying spacetime.

\subsection{Cartan geometry on observer space}\label{subsec:obscartan}
We start our discussion of observer space Cartan geometries with a heuristic construction of the model Klein geometry, following~\cite{Hohmann:2013fca}. Recall that the Cartan geometries displayed in the previous sections were defined on bundles \(\pi: P \to M\), where we identified the base manifold \(M\) with spacetime and the total space \(P\) with a suitably chosen admissible frame bundle. The first-order reductive model Klein geometry \(G/H\) with Lie algebra \(\mathfrak{g} = \mathfrak{h} \oplus \mathfrak{z}\) could then be interpreted such that \(\mathfrak{h}\) describes infinitesimal transformations of an admissible frame at a fixed point of spacetime, while \(\mathfrak{z}\) describes infinitesimal translations.

The basic idea behind our construction of an observer space model is to use the same interpretation for the total space \(P\) of admissible frames and the same Lie group \(G\) as for the spacetime model, but to choose a smaller subgroup \(K \subset H \subset G\), which keeps not only the observer's position fixed, but also one component of the observer's frame. This component can then be interpreted as the time component of the frame, and thus as the observer's velocity, while the remaining components are purely spatial. The action of \(K\) on \(P\) transforms only these spatial components and can thus be interpreted as a spatial rotation. The orbits of this group action form the fibers of a fiber bundle \(\hat{\pi}: P \to O\), whose base space \(O\) we call \emph{observer space}. Its geometry can be described by a Cartan geometry modeled on \(G/K\). The aim of this section is to introduce a class of model Klein geometries \(G/K\) which allow such an interpretation as observer space models, but to define them without referring to the underlying spacetime manifold \(M\). For this purpose we will employ the following definition.

\begin{defi}[Observer space model geometry]
Let \(G\) be a Lie group and \(K \subset G\) a closed subgroup. We call the Klein geometry \(G/K\) an \emph{observer space model} if there exists a split
\begin{equation}\label{eqn:obsalgsplit}
\mathfrak{g} = \mathfrak{k} \oplus \mathfrak{y} \oplus \vec{\mathfrak{z}} \oplus \mathfrak{z}^0
\end{equation}
into subrepresentations of the adjoint representation of \(K \subset G\) on \(\mathfrak{g}\) such that \(\mathfrak{z}^0\) is one-dimensional and \([\mathfrak{z}^0,\mathfrak{y}] = \vec{\mathfrak{z}}\), with the map \(\mathfrak{y} \to \vec{\mathfrak{z}}, y \mapsto [z,y]\) for fixed, non-zero \(z \in \mathfrak{z}^0\) being an isomorphism of representations of \(K\).
\end{defi}

It follows already from our definition that a Cartan geometry modeled on \(G/K\) is reductive. In the following we further restrict ourselves to first-order Cartan geometries, which implies that the adjoint representation of \(K \subset G\) on \(\mathfrak{y}\) and \(\vec{\mathfrak{z}}\) must be faithful. Note that we do not require the split~\eqref{eqn:obsalgsplit} to be irreducible, as this will not be necessary for our purposes; however, for the example we discuss in section~\ref{sec:ortobs} it will indeed be irreducible.

Recall from the definition in section~\ref{subsec:cartan} that for a Cartan geometry \((\hat{\pi}: P \to O, A)\) the restrictions \(A_p: T_pP \to \mathfrak{g}\) and \(\underline{A}_p: \mathfrak{g} \to T_pP\) at each point \(p \in P\) are isomorphisms of vector spaces. The split~\eqref{eqn:obsalgsplit} therefore induces a similar split of each tangent space, and hence a split of the tangent bundle \(TP\). Further, since the split of \(\mathfrak{g}\) is invariant under the adjoint representation of \(K\) and \(A\) is $K$-equivariant, the split of \(TP\) is invariant under the action of \(K\) on \(P\). Thus, it induces a similar split of \(TO\). This is summarized by the following diagram for every \(p \in P\).

\begin{equation}\label{eqn:tangbundsplit}
\xymatrix{\mathfrak{k} \ar@{^(->>}[d]_{\underline{A}_p} \ar@{}[r]|{\oplus} & \mathfrak{y} \ar@{^(->>}[d]_{\underline{A}_p} \ar@{}[r]|{\oplus} & \vec{\mathfrak{z}} \ar@{^(->>}[d]_{\underline{A}_p} \ar@{}[r]|{\oplus} & \mathfrak{z}^0 \ar@{^(->>}[d]_{\underline{A}_p} \ar@{}[r]|{=} & \mathfrak{g} \ar@{^(->>}[d]_{\underline{A}_p}\\
R_pP \ar[d]_{\hat{\pi}_*} \ar@{}[r]|{\oplus} & B_pP \ar@{^(->>}[d]_{\hat{\pi}_*} \ar@{}[r]|{\oplus} & \vec{H}_pP \ar@{^(->>}[d]_{\hat{\pi}_*} \ar@{}[r]|{\oplus} & H^0_pP \ar@{^(->>}[d]_{\hat{\pi}_*} \ar@{}[r]|{=} & T_pP \ar[d]_{\hat{\pi}_*}\\
0 & V_{\hat{\pi}(p)}O \ar@{}[r]|{\oplus} & \vec{H}_{\hat{\pi}(p)}O \ar@{}[r]|{\oplus} & H^0_{\hat{\pi}(p)}O \ar@{}[r]|{=} & T_{\hat{\pi}(p)}O}
\end{equation}

Arrows of the shape \(\xymatrix{\ar@{^(->>}[r]&}\) indicate that the corresponding linear maps are vector space isomorphisms. It thus follows in particular that the subspaces of \(T_pP\) have the same dimensions as the corresponding subspaces of \(\mathfrak{g}\) at each point \(p\), so that the subbundles of \(TP\) are of constant rank. Note also that the subbundle \(RP\) agrees with the vertical tangent bundle of the bundle \(\hat{\pi}: P \to O\), i.e., the kernel of \(\hat{\pi}_*: TP \to TO\). We denote the projections to the subbundles \(VO\), \(\vec{H}O\) and \(H^0O\) by
\begin{equation}
\Pi_V\,, \quad \Pi_{\vec{H}}\,, \quad \Pi_{H^0}\,, \quad \Pi_H = \Pi_{\vec{H}} + \Pi_{H^0}\,.
\end{equation}
Since \(\mathfrak{z}^0\) is one-dimensional by our definition above, we can pick a non-zero element \(\mathcal{Z}_0 \in \mathfrak{z}^0\) which is unique up to rescaling. Via the map \(\hat{\pi}_* \circ \underline{A}\) this element induces a section \(\mathbf{r}\) of \(H^0O\), which we call the \emph{Reeb vector field}. Further, it yields an isomorphism \([\mathcal{Z}_0, \bullet]: \mathfrak{y} \to \vec{\mathfrak{z}}\) of representations of \(K\). This isomorphism finally yields a vector bundle isomorphism \(\Theta: VO \to \vec{H}O\).

We conclude this section by providing a physical interpretation for the split~\eqref{eqn:obsalgsplit} and the objects derived from our definition of an observer space model. The Lie algebra \(\mathfrak{k} \subset \mathfrak{g}\) corresponds to infinitesimal transformations of the spatial frame components and thus takes the role of the algebra of infinitesimal rotations. The corresponding bundle \(RP\) consists of tangent vectors to the fibers of the bundle \(\hat{\pi}: P \to O\). The remaining subspaces of \(\mathfrak{g}\), and thus the corresponding subbundles of \(TP\) and \(TO\) are divided into boosts \(\mathfrak{y}\), which change the observer's velocity, and translations \(\mathfrak{z}\), which change his position and further split into a spatial part \(\vec{\mathfrak{z}}\) and a temporal part \(\mathfrak{z}^0\), relative to the observer's velocity. Note that if the vertical subbundle \(VO\) is an integrable distribution, it can be integrated to a foliation, whose leaf space can be identified with spacetime \(M\). The physical interpretation behind this case is that two observers \(o,o' \in O\) can be regarded as being at the same point \(x \in M\) of spacetime (although possibly having different velocities), if and only if they belong to the same leaf \(x\) of the foliation of \(O\). However, we will not need this condition for our constructions in this article, and thus also allow observer spaces for which \(VO\) is not integrable. In this case, there is no physical notion of (absolute) spacetime, as there is no equivalence relation of ``being at the same point'' between observers \(o,o' \in O\); however, one may define a notion of relative or local spacetime~\cite{Gielen:2012fz}.

A special role is assigned to the space \(\mathfrak{z}^0\), where \(\mathcal{Z}_0\) is interpreted as the generator of time translation. The vector fields \(\underline{A}(\mathcal{Z}_0) \in \Vect(P)\) and \(\mathbf{r} \in \Vect(O)\) can be interpreted to govern the time evolution of observers and their frames. The interpretation of \(\mathcal{Z}_0\) as the generator of time translation also enters the interpretations of \(\mathfrak{y}\) and \(\vec{\mathfrak{z}}\), as well as the corresponding subbundles of \(TP\) and \(TO\), as boosts and spatial translations. The condition that \([\mathcal{Z}_0, \bullet]: \mathfrak{y} \to \vec{\mathfrak{z}}\) is an isomorphism means that the time evolution of an infinitesimal boost (i.e., a change of the observer's velocity) generates an infinitesimal spatial translation (i.e., a change in the observer's position). This appears to be a natural condition, and will in fact be necessary for our discussion of symmetries, as we will see in the following section.

\subsection{Symmetries for observer spaces}\label{subsec:obssym}
We now consider symmetries of Cartan geometries \((\hat{\pi}: P \to O, A)\) based on the model Klein geometry \(G/K\) introduced in the preceding section. Since we are now dealing with a first-order reductive Cartan geometry over the observer space \(O\), also generators of symmetries will be vector fields \(\Xi \in \Vect(O)\). Any such vector field can of course be lifted to \(\mathrm{GL}(O)\) as shown in section~\ref{subsec:lifts}, and thus be used to define symmetries as shown in section~\ref{subsec:cartansym}, where we now view \(P\) as a subbundle of \(\mathrm{GL}(O)\). However, we wish to focus on vector fields \(\Xi\) which can be interpreted by an observer as generators of spacetime transformations. The following definition, which does not make any reference to spacetime itself, but uses only the elements of observer space Cartan geometry detailed in the previous section, is one of the central notions proposed in this article.

\begin{defi}[Spatio-temporal vector field]
We call a vector field \(\Xi\) on \(O\) \emph{spatio-temporal} if its boost component is the time derivative of the translation component,
\begin{equation}
\Pi_H\mathcal{L}_{\mathbf{r}}(\Pi_H\Xi) = \Theta\Pi_V\Xi\,,
\end{equation}
and the translation component does not change along boost directions,
\begin{equation}
\Pi_H\mathcal{L}_{\Upsilon}(\Pi_H\Xi) = 0
\end{equation}
for all \(\Upsilon \in \Gamma(VO)\).
\end{defi}

These two conditions can be interpreted as follows. The vector field \(\Xi\) splits into horizontal and vertical components \(\Pi_H\Xi\) and \(\Pi_V\Xi\), which change the observer's position and velocity, respectively. The Lie derivative \(\mathcal{L}_{\mathbf{r}}(\Pi_H\Xi) = [\mathbf{r},\Pi_H\Xi]\) consists of two components: a horizontal component measuring the change of \(\Pi_H\Xi\) along \(\mathbf{r}\) and a vertical component measuring the non-integrability of the horizontal distribution \(HO\). Here we are only interested in the first part, and thus apply the horizontal projection again. The result can be seen as the time derivative of the change of the observer's position. We compare it to the change of the observer's velocity, which is modeled by the vertical component \(\Pi_V\Xi\). For this purpose we use the bundle isomorphism \(\Theta: VO \to \vec{H}O\). Note that we have applied the projection \(\Pi_H\) instead of \(\Pi_{\vec{H}}\) on the left hand side, implying that the time translation component of \(\mathcal{L}_{\mathbf{r}}(\Pi_H\Xi)\) should vanish.

The second condition has a similar interpretation. The Lie derivative \(\mathcal{L}_{\Upsilon}(\Pi_H\Xi)\) splits into a horizontal part measuring the change of \(\Pi_H\Xi\) along \(\Upsilon\) and a vertical part measuring the converse. Here we are only interested in the horizontal part. The condition then expresses that observers undergo the same change of position by \(\Pi_H\Xi\) if they are related to each other by a pure boost. This condition becomes more clear if the vertical distribution \(VO\) can be integrated to a foliation, whose leaf space is a smooth manifold \(M\), with a surjective submersion \(\bar{\pi}: O \to M\). In this case the vector field \(\Upsilon\) generates a leaf-preserving one-parameter group \(t \mapsto \upsilon_t\) of diffeomorphisms of \(O\). The second condition then simply translates to
\begin{equation}
\bar{\pi}_* \circ \Pi_H\Xi = \bar{\pi}_* \circ \upsilon_{t*} \circ \Pi_H\Xi\,,
\end{equation}
taking into account that the kernel of \(\bar{\pi}_*\) is the vertical subbundle. This means that \(\bar{\pi}_* \circ \Pi_H\Xi: O \to TM\) is constant on the leaves of the foliation, so that it defines a vector field on \(M\), which is interpreted as a generator of transformations of the spacetime manifold. However, we emphasize again that we do not restrict our definition to the case of integrable vertical bundles, which is the reason for using the Lie derivative to restrict only the local properties of \(\Xi\).

This concludes our discussion of general observer space Cartan geometries and their symmetries. The general construction detailed in this section will be applied to a particular example in the following section, which also serves as a further illustration.

\section{Orthogonal observer space Cartan geometries}\label{sec:ortobs}
We now come to a particular class of observer space Cartan geometries in the sense of the definition detailed in the previous section. This class relates to the general observer space model in the same way as the orthogonal model discussed in section~\ref{sec:orthogonal} relates to a general Cartan geometry of spacetime, and will therefore also be called orthogonal. We start with a brief discussion of the Lie groups used in the orthogonal observer space model and their Lie algebra structure in section~\ref{subsec:obsort}. We then briefly review how a Cartan geometry based on this model can be derived from a Finsler spacetime geometry in section~\ref{subsec:finslercartan}. We finally discuss the symmetries of Finsler spacetimes in section~\ref{subsec:finslersym}.

\subsection{The orthogonal observer space model}\label{subsec:obsort}
The model Klein geometry we discuss in this section is closely related to the orthogonal model introduced in section~\ref{subsec:ortklein}, where the Klein geometry \(G/H\) could be interpreted as a maximally symmetric space equipped with a metric of signature \((m,n)\). Here we restrict ourselves to metrics with \(m = 1\) time dimension and \(n = d\) spatial dimensions. Further, we do not divide \(G\) by the full Lorentz group \(H\), but only by the group \(K\) of spatial rotations. This model Klein geometry can thus be summarized as follows.
\begin{equation}\label{eqn:ortobsgroups}
G = \begin{cases}
\mathrm{O}(1,d+1) & \text{for } \Lambda = 1\\
\mathrm{IO}(1,d) & \text{for } \Lambda = 0\\
\mathrm{O}(2,d) & \text{for } \Lambda = -1
\end{cases}\,, \quad H = \mathrm{O}(1,d)\,, \quad K = \mathrm{O}(d)\,.
\end{equation}
Here \(G\) can be any of the three groups listed above, where we introduced a parameter \(\Lambda \in \{-1, 0, 1\}\) indicating the scalar curvature of the corresponding Klein geometry \(G/H\), which can be interpreted as the cosmological constant on this maximally symmetric spacetime~\cite{Gielen:2012fz}. The Lie algebra \(\mathfrak{g}\) splits in the form~\eqref{eqn:obsalgsplit} into subrepresentations of the adjoint representation of \(K\) on \(\mathfrak{g}\). This split is related to the split \(\mathfrak{g} = \mathfrak{h} \oplus \mathfrak{z}\) by \(\mathfrak{h} = \mathfrak{k} \oplus \mathfrak{y}\) and \(\mathfrak{z} = \vec{\mathfrak{z}} \oplus \mathfrak{z}^0\). This allows us to write any element \(a \in \mathfrak{g}\) in the form
\begin{equation}
a = \frac{1}{2}k^{\mu}{}_{\nu}\mathcal{K}_{\mu}{}^{\nu} + y^{\mu}\mathcal{Y}_{\mu} + z^{\mu}\mathcal{Z}_{\mu} + z^0\mathcal{Z}_0\,,
\end{equation}
where the generators of \(\mathfrak{g}\) satisfy the Lie algebra relations
\begin{gather}
[\mathcal{K}_{\mu}{}^{\nu},\mathcal{K}_{\rho}{}^{\sigma}] = \delta^{\nu}_{\rho}\mathcal{K}_{\mu}{}^{\sigma} - \delta^{\sigma}_{\mu}\mathcal{K}_{\rho}{}^{\nu} + \delta_{\mu\rho}\delta^{\sigma\tau}\mathcal{K}_{\tau}{}^{\nu} - \delta^{\nu\sigma}\delta_{\rho\tau}\mathcal{K}_{\mu}{}^{\tau}\,, \quad [\mathcal{K}_{\mu}{}^{\nu},\mathcal{Z}_0] = 0\,,\nonumber\\
[\mathcal{K}_{\mu}{}^{\nu},\mathcal{Y}_{\rho}] = \delta^{\nu}_{\rho}\mathcal{Y}_{\mu} - \delta_{\mu\rho}\delta^{\nu\sigma}\mathcal{Y}_{\sigma}\,, \quad [\mathcal{K}_{\mu}{}^{\nu},\mathcal{Z}_{\rho}] = \delta^{\nu}_{\rho}\mathcal{Z}_{\mu} - \delta_{\mu\rho}\delta^{\nu\sigma}\mathcal{Z}_{\sigma}\,, \quad [\mathcal{Y}_{\mu},\mathcal{Y}_{\nu}] = -\delta_{\mu\rho}\mathcal{K}_{\nu}{}^{\rho}\,,\label{eqn:obsortalgebra}\\
[\mathcal{Y}_{\mu},\mathcal{Z}_{\nu}] = \delta_{\mu\nu}\mathcal{Z}_0\,, \quad [\mathcal{Y}_{\mu},\mathcal{Z}_0] = \mathcal{Z}_{\mu}\,, \quad [\mathcal{Z}_{\mu},\mathcal{Z}_{\nu}] = \Lambda\delta_{\mu\rho}\mathcal{K}_{\nu}{}^{\rho}\,, \quad [\mathcal{Z}_{\mu},\mathcal{Z}_0] = \Lambda\mathcal{Y}_{\mu}\,.\nonumber
\end{gather}
These depend on the parameter \(\Lambda\), and thus on the particular choice of a model geometry.

In order to discuss Cartan geometries based on these models, we first make a few remarks on the structure of the admissible frame bundle, and apply the same construction as we did also for the orthogonal spacetime geometries in section~\ref{subsec:ortklein}. Note that the space \(\mathfrak{y} \oplus \vec{\mathfrak{z}} \oplus \mathfrak{z}^0\) is equipped with a metric \(\eta\) of signature \((1,2d)\), such that \(\eta(\mathcal{Z}_i,\mathcal{Z}_j) = \eta_{ij}\), \(\eta(\mathcal{Y}_{\mu},\mathcal{Y}_{\nu}) = \delta_{\mu\nu}\) and \(\eta(\mathcal{Z}_i,\mathcal{Y}_{\mu}) = 0\), which is invariant under the adjoint representation of \(K \subset G\). An admissible frame \(p \in P_o \subset P\) over \(o \in O\), which is a linear map \(p: \mathfrak{y} \oplus \vec{\mathfrak{z}} \oplus \mathfrak{z}^0 \to T_oO\), thus defines a metric \(g(.,.) = \eta(p^{-1}(.), p^{-1}(.))\) of the same signature \((1,2d)\) on \(O\). This metric is independent of the choice of the frame \(p \in P_o\).

Recall further from section~\ref{subsec:obscartan} that an observer space Cartan geometry is equipped with a unique split \(TO = VO + \vec{H}O + H^0O\) of the tangent bundle. The admissible frame bundle thus turns out to be constituted by frames which respect the split of the tangent bundle, and which are orthonormal with respect to the metric \(g\).

To further discuss objects on the frame bundle, we introduce coordinates first on \(\mathrm{GL}(O)\) and then restrict them to \(P\). Given coordinates \((w^A)\) on \(O\), we can introduce coordinates \((w^A,v_i^A,\underline{v}_{\mu}^A)\) on \(\mathrm{GL}(O)\) such that a linear map \(v: \mathfrak{y} \oplus \vec{\mathfrak{z}} \oplus \mathfrak{z}^0 \to T_oO\) takes the form
\begin{equation}
v(\mathcal{Y}_{\mu}) = \underline{v}_{\mu}^A\partial_A\,, \quad v(\mathcal{Z}_i) = v_i^A\partial_A\,.
\end{equation}
Using these coordinates, restricted to \(P\), we can now discuss the general properties of the Cartan connection. First of all, we use the generators~\eqref{eqn:obsortalgebra} to write the Cartan connection as
\begin{equation}\label{eqn:ortobscartconn}
A = \Omega + b + \vec{e} + e^0 = \frac{1}{2}\Omega^{\mu}{}_{\nu}\mathcal{K}_{\mu}{}^{\nu} + b^{\mu}\mathcal{Y}_{\mu} + e^{\mu}\mathcal{Z}_{\mu} + e^0\mathcal{Z}_0\,,
\end{equation}
while for the associated vector fields we write
\begin{equation}\label{eqn:ortobsfundvect}
\underline{A}\left(\frac{1}{2}k^{\mu}{}_{\nu}\mathcal{K}_{\mu}{}^{\nu} + y^{\mu}\mathcal{Y}_{\mu} + z^{\mu}\mathcal{Z}_{\mu} + z^0\mathcal{Z}_0\right) = \frac{1}{2}k^{\mu}{}_{\nu}\underline{\Omega}_{\mu}{}^{\nu} + y^{\mu}\underline{b}_{\mu} + z^{\mu}\underline{e}_{\mu} + z^0\underline{e}_0\,.
\end{equation}
The one-forms \(\mu_o \in \Omega^1(P_o, \mathfrak{k})\) induced by the Maurer-Cartan form \(\mu \in \Omega^1(K, \mathfrak{k})\) on the fibers of \(P\) and the canonical vector fields \(\widetilde{\mathcal{K}}_{\mu}{}^{\nu}\) of the generators \(\mathcal{K}_{\mu}{}^{\nu}\) are given by
\begin{equation}\label{eqn:ortobsmccan}
\mu_o = \frac{1}{2}v^{-1}{}^{\mu}_Adv_{\nu}^A\mathcal{K}_{\mu}{}^{\nu} = \frac{1}{2}\underline{v}^{-1}{}^{\mu}_Ad\underline{v}_{\nu}^A\mathcal{K}_{\mu}{}^{\nu} \quad \text{and} \quad \widetilde{\mathcal{K}}_{\mu}{}^{\nu} = v_{\mu}^A\bar{\partial}^{\nu}_A + \underline{v}_{\mu}^A\underline{\bar{\partial}}^{\nu}_A - \delta_{\mu\rho}\delta^{\nu\sigma}(v_{\sigma}^A\bar{\partial}^{\rho}_A + \underline{v}_{\sigma}^A\underline{\bar{\partial}}^{\rho}_A)\,.
\end{equation}
Note that we have provided two different coordinate expressions for \(\mu_o\). Although they define different one-forms on \(\mathrm{GL}(O)\), either of them restricts to the same one-form on \(P\). These coordinate expressions are also consistent with the canonical vector fields, which are related to the one-forms by \(\mu_o(\widetilde{\mathcal{K}}_{\mu}{}^{\nu}) = \mathcal{K}_{\mu}{}^{\nu}\). The particular form of the canonical vector fields given above comes from the fact that rotations act both on the $\mathfrak{y}$-components \(\underline{v}_{\mu}^A\) and the $\vec{z}$-components \(v_{\mu}^A\) of a frame \(v\). Making use of these expressions, we can write the conditions~\ref{cartan:mcform} and~\ref{cartan:canonical} in the form
\begin{equation}\label{eqn:ortobsmccan2}
e^i(\underline{\Omega}_{\mu}{}^{\nu}) = 0\,, \quad b^{\rho}(\underline{\Omega}_{\mu}{}^{\nu}) = 0\,, \quad \Omega^{\rho}{}_{\sigma}(\underline{\Omega}_{\mu}{}^{\nu}) = \delta^{\rho}_{\mu}\delta^{\nu}_{\sigma} - \eta_{\nu\rho}\eta^{\mu\sigma}
\end{equation}
and
\begin{equation}\label{eqn:ortobsmccan3}
\underline{\Omega}_{\mu}{}^{\nu} = v_{\mu}^A\bar{\partial}^{\nu}_A + \underline{v}_{\mu}^A\underline{\bar{\partial}}^{\nu}_A - \delta_{\mu\rho}\delta^{\nu\sigma}(v_{\sigma}^A\bar{\partial}^{\rho}_A + \underline{v}_{\sigma}^A\underline{\bar{\partial}}^{\rho}_A)\,.
\end{equation}
Finally, the conditions~\ref{cartan:equivariant} and~\ref{cartan:equivariant2} in differential form are given by
\begin{gather}
\mathcal{L}_{\widetilde{\mathcal{K}}_{\mu}{}^{\nu}}\Omega = -\frac{1}{2}\left(\Omega^{\nu}{}_{\rho}\mathcal{K}_{\mu}{}^{\rho} - \Omega^{\rho}{}_{\mu}\mathcal{K}_{\rho}{}^{\nu} + \Omega^{\rho}{}_{\sigma}\eta_{\mu\rho}\eta^{\sigma\tau}\mathcal{K}_{\tau}{}^{\nu} - \Omega^{\rho}{}_{\sigma}\eta^{\nu\sigma}\eta_{\rho\tau}\mathcal{K}_{\mu}{}^{\tau}\right)\,, \quad \mathcal{L}_{\widetilde{\mathcal{K}}_{\mu}{}^{\nu}}e^0 = 0\,,\nonumber\\
\mathcal{L}_{\widetilde{\mathcal{K}}_{\mu}{}^{\nu}}b = -b^{\nu}\mathcal{Y}_{\mu} + b^{\rho}\eta_{\mu\rho}\eta^{\nu\sigma}\mathcal{Y}_{\sigma}\,, \quad \mathcal{L}_{\widetilde{\mathcal{K}}_{\mu}{}^{\nu}}\vec{e} = -e^{\nu}\mathcal{Z}_{\mu} + e^{\rho}\eta_{\mu\rho}\eta^{\nu\sigma}\mathcal{Z}_{\sigma}\,,\label{eqn:ortobseqv}
\end{gather}
and
\begin{gather}
\mathcal{L}_{\widetilde{\mathcal{K}}_{\mu}{}^{\nu}}\underline{\Omega}_{\rho}{}^{\sigma} = \delta^{\nu}_{\rho}\underline{\Omega}_{\mu}{}^{\sigma} - \delta^{\sigma}_{\mu}\underline{\Omega}_{\rho}{}^{\nu} + \delta_{\mu\rho}\delta^{\sigma\tau}\underline{\Omega}_{\tau}{}^{\nu} - \delta^{\nu\sigma}\delta_{\rho\tau}\underline{\Omega}_{\mu}{}^{\tau}\,, \quad \mathcal{L}_{\widetilde{\mathcal{K}}_{\mu}{}^{\nu}}\underline{e}_0 = 0\,,\nonumber\\
\mathcal{L}_{\widetilde{\mathcal{K}}_{\mu}{}^{\nu}}\underline{b}_{\rho} = \delta^{\nu}_{\rho}\underline{b}_{\mu} - \delta_{\mu\rho}\delta^{\nu\sigma}\underline{b}_{\sigma}\,, \quad \mathcal{L}_{\widetilde{\mathcal{K}}_{\mu}{}^{\nu}}\underline{e}_{\rho} = \delta^{\nu}_{\rho}\underline{e}_{\mu} - \delta_{\mu\rho}\delta^{\nu\sigma}\underline{e}_{\sigma}\,.\label{eqn:ortobseqv2}
\end{gather}
Here we used the fact that the vertical tangent bundle of \(P\) is spanned by the canonical vector fields \(\widetilde{\mathcal{K}}_{\mu}{}^{\nu}\). With these properties at hand, we can now use the model geometry discussed above in order to describe the geometry of Finsler spacetimes in the next section.

\subsection{Finsler spacetimes in Cartan language}\label{subsec:finslercartan}
We now express the Finsler spacetime geometry detailed in section~\ref{subsec:finsler} as an observer space Cartan geometry modeled on \(G/K\) with \(G\) and \(K\) given in the definition~\eqref{eqn:ortobsgroups}. For this purpose we follow the construction detailed in~\cite{Hohmann:2013fca}, which we briefly summarize here. We first define the observer space \(O\) as the union of the future unit timelike shells defined in condition~\ref{finsler:timelike} in section~\ref{subsec:finsler}, whose elements correspond to the physical observer velocities,
\begin{equation}\label{eqn:finslerobsspace}
O = \bigcup_{x \in M}S_x \subset TM\,.
\end{equation}
We then define a principal $K$-bundle \(\pi: P \to O\). It is most convenient to first define it as a subbundle of the frame bundle \(\mathrm{GL}(M)\), which we then identify with a subbundle of \(\mathrm{GL}(O)\) following the construction shown in section~\ref{subsec:forcartan}. Using the coordinates \((x^a, f_i^a)\) on the frame bundle introduced in section~\ref{subsec:lifts} we define
\begin{equation}\label{eqn:finslerobsbundle}
P = \left\{(x,f) \in \mathrm{GL}(M)\,|\,f_0 \in O \; \text{and} \; g^F_{ab}(x,f_0)f_i^af_j^b = -\eta_{ij}\right\}
\end{equation}
with the canonical projection \(\pi(x,f) = (x,f_0)\). Here \(g^F_{ab}\) denotes the Finsler metric. Using the split~\eqref{eqn:ortobscartconn} we can write the Cartan connection \(A \in \Omega^1(P, \mathfrak{g})\) in the form
\begin{equation}
\Omega^{\mu}{}_{\nu} = f^{-1}{}^{\mu}_a\left[df_{\nu}^a + f_{\nu}^b\left(F^a{}_{bc}dx^c + C^a{}_{bc}\delta y^c\right)\right]\,, \quad b^{\mu} = f^{-1}{}^{\mu}_a\delta y^a\,, \quad e^i = f^{-1}{}^i_adx^a\,,
\end{equation}
where we wrote \(y^a\) synonymously for the time component \(f_0\) of the frame \(f\), which corresponds to the observer's velocity, and \(\delta y^a\) refers to the Berwald basis~\eqref{eqn:cotberwald}. \(F^a{}_{bc}\) and \(C^a{}_{bc}\) denote the coefficients of the Cartan linear connection~\eqref{eqn:cartlinconn}. The corresponding associated vector fields take the form
\begin{equation}
\underline{\Omega}_{\mu}{}^{\nu} = f_{\mu}^a\bar{\partial}^{\nu}_a - \delta_{\mu\rho}\delta^{\nu\sigma}f_{\sigma}^a\bar{\partial}^{\rho}_a\,, \quad \underline{b}_{\mu} = f_{\mu}^a\left(\bar{\partial}^0_a - f_j^bC^c{}_{ab}\bar{\partial}^j_c\right)\,, \quad \underline{e}_i = f_i^a\left(\partial_a - f_j^bF^c{}_{ab}\bar{\partial}^j_c\right)\,.
\end{equation}
In order to identify \(P\) with a subbundle of \(\mathrm{GL}(O)\) we first introduce a suitable set of coordinates. Let \((\hat{x}^a,u^{\alpha})\) be coordinates on \(O\) such that \(x^a = \hat{x}^a\) and \(y^a = y^a(\hat{x}^a, u^{\alpha})\). We can identify a frame at \(o \in O\) with a linear map \(v: \mathfrak{y} \oplus \vec{\mathfrak{z}} \oplus \mathfrak{z}^0 \to T_oO\) and introduce coordinates \((\hat{x}^a,u^{\alpha},\hat{v}_i^a,\tilde{v}_i^{\alpha},\underline{\hat{v}}_{\mu}^a,\underline{\tilde{v}}_{\mu}^{\alpha})\) on \(\mathrm{GL}(O)\) such that
\begin{equation}
v(\mathcal{Y}_{\mu}) = \underline{\hat{v}}_{\mu}^a\hat{\partial}_a + \underline{\tilde{v}}_{\mu}^{\alpha}\tilde{\partial}_{\alpha}\,, \quad v(\mathcal{Z}_i) = \hat{v}_i^a\hat{\partial}_a + \tilde{v}_i^{\alpha}\tilde{\partial}_{\alpha}\,,
\end{equation}
where \(\hat{\partial}_a\) and \(\tilde{\partial}_{\alpha}\) are the coordinate vector fields corresponding to the coordinates \((\hat{x}^a,u^{\alpha})\) on \(O\). Similarly, we write the inverse frames as
\begin{equation}
v^{-1}(\hat{\partial}_a) = \underline{\hat{v}}^{-1}{}^{\mu}_a\mathcal{Y}_{\mu} + \hat{v}^{-1}{}^i_a\mathcal{Z}_i\,, \quad v^{-1}(\tilde{\partial}_{\alpha}) = \underline{\tilde{v}}^{-1}{}^{\mu}_{\alpha}\mathcal{Y}_{\mu} + \tilde{v}^{-1}{}^i_{\alpha}\mathcal{Z}_i\,.
\end{equation}
In these coordinates we can identify a frame \((x,f) \in P\) over \((x,f_0) = (\hat{x},u) \in O\) using the relations
\begin{equation}\label{eqn:obsfbembed}
\hat{v}_i^a = f_i^a\,, \quad \tilde{v}_i^{\alpha} = f_i^a\delta_au^{\alpha}\,, \quad \underline{\hat{v}}_{\mu}^a = 0\,, \quad \underline{\tilde{v}}_{\mu}^{\alpha} = f_{\mu}^a\bar{\partial}_au^{\alpha}\,.
\end{equation}
For later use we also provide the components of the inverse frames \(v^{-1}\) in the same basis, which are given by
\begin{equation}\label{eqn:invobsfbembed}
\hat{v}^{-1}{}^i_a = f^{-1}{}^i_a\,, \quad \tilde{v}^{-1}{}^i_{\alpha} = 0\,, \quad \underline{\hat{v}}^{-1}{}^{\mu}_a = f^{-1}{}^{\mu}_b\tilde{\partial}_{\alpha}y^b\delta_au^{\alpha}\,, \quad \underline{\tilde{v}}^{-1}{}^{\mu}_{\alpha} = f^{-1}{}^{\mu}_a\tilde{\partial}_{\alpha}y^a\,.
\end{equation}
As a subbundle of \(\mathrm{GL}(O)\) we can then express \(P\) as
\begin{multline}\label{eqn:finslerobsbundle2}
P = \big\{(\hat{x},u,v) \in \mathrm{GL}(O)\,|\,g^F_{ab}(\hat{x},u)\hat{v}_i^a\hat{v}_j^b = -\eta_{ij}, \, \hat{v}_0^a = y^a(\hat{x},u), \, \hat{v}_{\mu}^a = \underline{\hat{v}}_{\mu}^b\hat{\partial}_by^a(\hat{x},u) + \underline{\tilde{v}}_{\mu}^{\alpha}\tilde{\partial}_{\alpha}y^a(\hat{x},u),\\
\underline{\hat{v}}_{\mu}^a = 0 \; \text{and} \; \hat{v}_i^b\hat{\partial}_by^a(\hat{x},u) + \tilde{v}_i^{\alpha}\tilde{\partial}_{\alpha}y^a(\hat{x},u) + N^a{}_b(\hat{x},u)\hat{v}_i^b = 0\big\}
\end{multline}
We could also express the Cartan connection \(A\) in terms of the new coordinates. However, it will be more convenient to use the coordinates \((x,f)\) on \(P\) and to use the newly introduced coordinates only for objects which are naturally defined on \(\mathrm{GL}(O)\).

\subsection{Symmetries for Finsler spacetimes}\label{subsec:finslersym}
We finally show how the relation between a Finsler spacetime and a Cartan geometry on its observer space can be used to translate the notions of symmetry between these two different geometries. Recall that in section~\ref{subsec:finsler} we called a vector field \(\xi\) on \(M\) a symmetry of the Finsler spacetime if its canonical lift \(\hat{\xi}\) to the tangent bundle \(TM\) leaves the fundamental geometry function \(L\) invariant, \(\mathcal{L}_{\hat{\xi}}L = 0\), essentially following the thorough discussion in~\cite{Tashiro:1959}. The following proposition connects this property to symmetries in the sense of Cartan geometry.

\begin{prop}
For a Finsler spacetime \((M,L,F)\) with induced observer space Cartan geometry \((\pi: P \to O, A)\) there is a one-to-one correspondence between vector fields \(\xi \in \Vect(M)\) inducing symmetries of the Finsler function and vector fields \(\Xi \in \Vect(O)\) leaving the Cartan geometry invariant, and these vector fields \(\Xi\) are spatio-temporal.
\end{prop}
\begin{proof}
The proof will proceed as follows. We will start from a vector field \(\Xi \in \Vect(O)\) and discuss under which circumstances its frame bundle lift \(\bar{\Xi} \in \Vect(\mathrm{GL}(O))\) is tangent to \(P\), which is a necessary condition for being a symmetry of the Cartan geometry. We will show that in this particular case it is also a sufficient condition, which means that a frame bundle lift \(\bar{\Xi}\) tangent to \(P\) also leaves the Cartan connection \(A \in \Omega^1(P, \mathfrak{g})\) invariant. We will then show that there exists a unique vector field \(\xi \in \Vect(M)\) inducing a symmetry of the Finsler function, whose tangent bundle lift \(\hat{\xi} \in \Vect(TM)\) is tangent to \(O\) and restricts to \(\Xi\) on \(O\). Conversely, we will show that the tangent bundle lift \(\hat{\xi}\) of any symmetry generator \(\xi\) of a Finsler spacetime restricts to a symmetry \(\Xi\) of the Cartan geometry on \(O\). We will finally show that \(\Xi\) is spatio-temporal. Note that it will turn out to be convenient to prove these statements in a different order, but the content of the proof will be the same as sketched here.

Let \(\Xi = \hat{X}^a\hat{\partial}_a + \tilde{X}^{\alpha}\tilde{\partial}_{\alpha}\) be a vector field on \(O\). Its lift \(\bar{\Xi}\) to the frame bundle \(\mathrm{GL}(O)\) takes the form
\begin{equation}
\begin{split}
\bar{\Xi} &= \hat{X}^a\hat{\partial}_a + \hat{v}_i^b\hat{\partial}_b\hat{X}^aD^i_a + \tilde{v}_i^{\beta}\tilde{\partial}_{\beta}\hat{X}^aD^i_a + \underline{\hat{v}}_{\mu}^b\hat{\partial}_b\hat{X}^a\overline{D}^{\mu}_a + \underline{\tilde{v}}_{\mu}^{\beta}\tilde{\partial}_{\beta}\hat{X}^a\overline{D}^{\mu}_a\\
&\phantom{=}+ \tilde{X}^{\alpha}\tilde{\partial}_{\alpha} + \hat{v}_i^b\hat{\partial}_b\tilde{X}^{\alpha}\underline{D}^i_{\alpha} + \tilde{v}_i^{\beta}\tilde{\partial}_{\beta}\tilde{X}^{\alpha}\underline{D}^i_{\alpha} + \underline{\hat{v}}_{\mu}^b\hat{\partial}_b\tilde{X}^{\alpha}\underline{\overline{D}}^{\mu}_{\alpha} + \underline{\tilde{v}}_{\mu}^{\beta}\tilde{\partial}_{\beta}\tilde{X}^{\alpha}\underline{\overline{D}}^{\mu}_{\alpha}\,,
\end{split}
\end{equation}
where \(\hat{\partial}_a, \tilde{\partial}_{\alpha}, D^i_a, \underline{D}^i_{\alpha}, \overline{D}^{\mu}_a, \underline{\overline{D}}^{\mu}_{\alpha}\) denote the coordinate vector fields corresponding to the coordinates \((\hat{x}^a,u^{\alpha},\hat{v}_i^a,\tilde{v}_i^{\alpha},\underline{\hat{v}}_{\mu}^a,\underline{\tilde{v}}_{\mu}^{\alpha})\). For \((x,u,v) \in P\), where these coordinates are given by the identification~\eqref{eqn:obsfbembed}, we thus find
\begin{equation}
\begin{split}
\bar{\Xi} &= \hat{X}^a\hat{\partial}_a + f_i^b\hat{\partial}_b\hat{X}^aD^i_a + f_i^b\delta_bu^{\beta}\tilde{\partial}_{\beta}\hat{X}^aD^i_a + f_{\mu}^b\bar{\partial}_bu^{\beta}\tilde{\partial}_{\beta}\hat{X}^a\overline{D}^{\mu}_a\\
&\phantom{=}+ \tilde{X}^{\alpha}\tilde{\partial}_{\alpha} + f_i^b\hat{\partial}_b\tilde{X}^{\alpha}\underline{D}^i_{\alpha} + f_i^b\delta_bu^{\beta}\tilde{\partial}_{\beta}\tilde{X}^{\alpha}\underline{D}^i_{\alpha} + f_{\mu}^b\bar{\partial}_bu^{\beta}\tilde{\partial}_{\beta}\tilde{X}^{\alpha}\underline{\overline{D}}^{\mu}_{\alpha}\\
&= \hat{X}^a\hat{\partial}_a + \tilde{X}^{\alpha}\tilde{\partial}_{\alpha} + f_i^b\delta_b\hat{X}^aD^i_a + f_{\mu}^b\bar{\partial}_b\hat{X}^a\overline{D}^{\mu}_a + f_i^b\delta_b\tilde{X}^{\alpha}\underline{D}^i_{\alpha} + f_{\mu}^b\bar{\partial}_b\tilde{X}^{\alpha}\underline{\overline{D}}^{\mu}_{\alpha}\,.
\end{split}
\end{equation}
In the following we rewrite the vector field \(\Xi\) in the form \(\Xi = X^a\partial_a + \bar{X}^a\bar{\partial}_a\), which will be more convenient for the remainder of the calculation. Note that the components in the different bases are related by
\begin{equation}\label{eqn:xbaserels}
\hat{X}^a = X^a\,, \quad \tilde{X}^{\alpha} = X^a\partial_au^{\alpha} + \bar{X}^a\bar{\partial}_au^{\alpha} \quad \Leftrightarrow \quad X^a = \hat{X}^a\,, \quad \bar{X}^a = \hat{X}^b\hat{\partial}_by^a + \tilde{X}^{\alpha}\tilde{\partial}_{\alpha}y^a\,.
\end{equation}
We now pose the question under which circumstances the frame bundle lift \(\bar{\Xi}\) is tangent to \(P \subset \mathrm{GL}(O)\). This is the case if and only if all of the following conditions are met, which are derived from the description~\eqref{eqn:finslerobsbundle2} of \(P\):
\begin{itemize}
\item
From the condition \(\underline{\hat{v}}_{\mu}^a = 0\) follows
\begin{equation}
0 = \bar{\Xi}(\underline{\hat{v}}_{\mu}^a) = f_{\mu}^b\bar{\partial}_b\hat{X}^a = f_{\mu}^b\bar{\partial}_bX^a\,.
\end{equation}
Since the vectors \(f_{\mu}^b\bar{\partial}_b\) for fixed \(f_{\mu}^b\) span a tangent space of the unit timelike shell \(S_x\), it follows that the components \(X^a\) must be constant on each such shell and may therefore depend only on the coordinates \((x^a)\) on \(M\). We will make use of this result to simplify the remaining conditions.

\item
From the condition \(\hat{v}_0^a = y^a(\hat{x},u)\) follows
\begin{equation}
0 = \bar{\Xi}[\hat{v}_0^a - y^a(\hat{x},u)] = f_0^b\delta_b\hat{X}^a - \hat{X}^b\hat{\partial}_by^a - \tilde{X}^{\alpha}\tilde{\partial}_{\alpha}y^a = y^b\partial_bX^a - \bar{X}^a\,.
\end{equation}
Here we have already replaced \(\delta_bX^a = \partial_bX^a\) using the preceding condition and \(f_0^a = y^a\). This condition now uniquely fixes \(\bar{X}^a\), and we will also use it as a simplification in the remaining conditions. Together with the first condition it also shows that \(\Xi = X^a\partial_a + y^b\partial_bX^a\bar{\partial}_a\) is the restriction of the tangent bundle lift \(\hat{\xi}\) of a vector field \(\xi = X^a\partial_a \in \Vect(M)\) to \(O\). Further, \(\hat{\xi}\) must be tangent to \(O\) by the definition of \(\Xi\), and must thus satisfy \(\mathcal{L}_{\hat{\xi}}F = \hat{\xi}(F) = 0\), since the Finsler function \(F\) is constant on \(O\). The same holds for the geometry function \(L\).

\item
From the condition \(\hat{v}_{\mu}^a = \underline{\hat{v}}_{\mu}^b\hat{\partial}_by^a(\hat{x},u) + \underline{\tilde{v}}_{\mu}^{\alpha}\tilde{\partial}_{\alpha}y^a(\hat{x},u)\) follows
\begin{equation}
\begin{split}
0 &= f_{\mu}^b\left[\bar{\partial}_b\hat{X}^c\hat{\partial}_cy^a + \bar{\partial}_b\tilde{X}^{\alpha}\tilde{\partial}_{\alpha}y^a + \bar{\partial}_bu^{\alpha}\left(\hat{X}^c\hat{\partial}_c\tilde{\partial}_{\alpha}y^a + \tilde{X}^{\beta}\tilde{\partial}_{\beta}\tilde{\partial}_{\alpha}y^a\right) - \delta_b\hat{X}^a\right]\\
&= f_{\mu}^b\left(\bar{\partial}_b\bar{X}^a - \delta_bX^a\right) = f_{\mu}^b\left(\partial_bX^a - \partial_bX^a\right)\,,
\end{split}
\end{equation}
where the second line can be derived from the first line by evaluating \(\bar{\partial}_b\bar{X}^a\) with the help of the last relation~\eqref{eqn:xbaserels}, and the last equality follows from the first two conditions, so that this condition is identically satisfied.

\item
From the condition \(g^F_{ab}(\hat{x},u)\hat{v}_i^a\hat{v}_j^b = -\eta_{ij}\) follows
\begin{equation}
\begin{split}
0 &= f_i^af_j^b\left(\hat{X}^c\hat{\partial}_cg^F_{ab} + \tilde{X}^{\alpha}\tilde{\partial}_{\alpha}g^F_{ab} + g^F_{cb}\delta_a\hat{X}^c + g^F_{ac}\delta_b\hat{X}^c\right)\\
&= f_i^af_j^b\left(X^c\partial_cg^F_{ab} + y^d\partial_dX^c\bar{\partial}_cg^F_{ab} + g^F_{cb}\partial_aX^c + g^F_{ac}\partial_bX^c\right) = f_i^af_j^b(\mathcal{L}_{\hat{\xi}}g^F)_{ab}\,.
\end{split}
\end{equation}
This is simply the change of the Finsler metric \(g^F_{ab}\) under a diffeomorphism of \(M\) generated by the vector field \(\xi\). This also vanishes as a consequence of the first two conditions, since
\begin{equation}
(\mathcal{L}_{\hat{\xi}}g^F)_{ab} = \frac{1}{2}\bar{\partial}_a\bar{\partial}_b\left(\mathcal{L}_{\hat{\xi}}F^2\right) = 0\,,
\end{equation}
using again the fact that the Finsler function is constant on \(O\).

\item
From the condition \(\hat{v}_i^b\hat{\partial}_by^a(\hat{x},u) + \tilde{v}_i^{\alpha}\tilde{\partial}_{\alpha}y^a(\hat{x},u) + N^a{}_b(\hat{x},u)\hat{v}_i^b = 0\) follows
\begin{equation}
0 = f_i^b\left(X^c\partial_cN^a{}_b + y^d\partial_dX^c\bar{\partial}_cN^a{}_b - \partial_cX^aN^c{}_b + \partial_bX^cN^a{}_c + y^c\partial_b\partial_cX^a\right) = f_i^b(\mathcal{L}_{\hat{\xi}}N)^a{}_b\,,
\end{equation}
which can be derived in analogy to the previous condition by evaluating \(\delta_b\bar{X}^a\) using the relations~\eqref{eqn:xbaserels}. The result we obtain here is the change of the coefficients \(N^a{}_b\) of the Cartan non-linear connection under a diffeomorphism of \(M\) generated by the vector field \(\xi\). This vanishes as a consequence of the first two conditions, since
\begin{equation}
(\mathcal{L}_{\hat{\xi}}N)^a{}_b = \frac{1}{4}\bar{\partial}_b\left\{\left(\mathcal{L}_{\hat{\xi}}g^F\right)^{ap}\left(y^q\partial_q\bar{\partial}_pF^2 - \partial_pF^2\right) + g^{F\,ap}\left[y^q\partial_q\bar{\partial}_p\left(\mathcal{L}_{\hat{\xi}}F^2\right) - \partial_p\left(\mathcal{L}_{\hat{\xi}}F^2\right)\right]\right\}\,,
\end{equation}
which vanishes by the same arguments as above.
\end{itemize}
The conditions derived above can be summarized as follows: The frame bundle lift \(\bar{\Xi}\) is tangent to \(P\) if and only if \(\Xi\) is the restriction to \(O\) of the tangent bundle lift \(\hat{\xi}\) of a vector field \(\xi \in \Vect(M)\) generating a symmetry of the Finsler spacetime.

We still need to check that the Cartan connection \(A\) is invariant under the restriction of \(\bar{\Xi}\) to \(P\). For the components \(b\) and \(e\) this is trivially satisfied, since
\begin{equation}
\begin{split}
b + e &= f^{-1}{}^{\mu}_a\tilde{\partial}_{\alpha}y^a\left(du^{\alpha} - \delta_bu^{\alpha}d\hat{x}^b\right)\mathcal{Y}_{\mu} + f^{-1}{}^i_ad\hat{x}^a\mathcal{Z}_i\\
&= \underline{\hat{v}}^{-1}{}^{\mu}_ad\hat{x}^a\mathcal{Y}_{\mu} + \underline{\tilde{v}}^{-1}{}^{\mu}_{\alpha}du^{\alpha}\mathcal{Y}_{\mu} + \hat{v}^{-1}{}^i_ad\hat{x}^a\mathcal{Z}_i + \tilde{v}^{-1}{}^i_{\alpha}du^{\alpha}\mathcal{Z}_i\,,
\end{split}
\end{equation}
where we made use of the inverse frame components~\eqref{eqn:invobsfbembed}. This is the restriction of the solder form on \(\mathrm{GL}(O)\) to \(P\), which is invariant under the frame bundle lift \(\bar{\Xi}\) of any vector field \(\Xi\) on \(O\). For \(\Omega^{\mu}{}_{\nu}\) we finally find
\begin{equation}
\mathcal{L}_{\bar{\Xi}}\Omega^{\mu}{}_{\nu} = f^{-1}{}^{\mu}_af_{\nu}^b\left[(\mathcal{L}_{\hat{\xi}}F)^a{}_{bc}dx^c + (\mathcal{L}_{\hat{\xi}}C)^a{}_{bc}\delta y^c + C^a{}_{bc}(\mathcal{L}_{\hat{\xi}}N)^c{}_ddx^d\right]\,.
\end{equation}
A lengthy but straightforward calculation shows that also \((\mathcal{L}_{\hat{\xi}}F)^a{}_{bc}\) and \((\mathcal{L}_{\hat{\xi}}C)^a{}_{bc}\) vanish due to the fact that \(\xi\) is a symmetry of the Finsler spacetime. This result completes the proof that any vector field \(\xi\) on \(M\) which leaves the Finsler geometry invariant uniquely corresponds to a vector field \(\Xi\) on \(O\) which leaves the Cartan geometry invariant and vice versa.

We finally show that \(\Xi\) is spatio-temporal, and thus can be interpreted by observers as the generator of a spacetime symmetry. For this purpose, note that on the observer space \(O\) of a Finsler spacetime the operators \(\Pi_H\), \(\Pi_V\) and \(\Theta\) take the form
\begin{equation}
\Pi_H = \delta_a \otimes dx^a\,, \quad \Pi_V = \bar{\partial}_a \otimes \delta y^a \quad \text{and} \quad \Theta = \delta_a \otimes \delta y^a\,,
\end{equation}
and that the Reeb vector field is given by \(\mathbf{r} = y^a\delta_a\). With
\begin{equation}
\Pi_H\Xi = X^a\delta_a \quad \text{and} \quad \Theta\Pi_V\Xi = (y^b\partial_bX^a + N^a{}_bX^b)\delta_a
\end{equation}
we thus find
\begin{multline}
\Pi_H\mathcal{L}_{\mathbf{r}}\Pi_H\Xi = \Pi_H[y^a\delta_a, X^b\delta_b] = \Pi_H(y^a\partial_aX^b\delta_b + X^bN^a{}_b\delta_a + y^aX^b[\delta_a, \delta_b])\\
= y^a\partial_aX^b\delta_b + X^bN^a{}_b\delta_a = \Theta\Pi_V\Xi\,.
\end{multline}
For a vertical vector field \(\Upsilon = \bar{Y}^a\bar{\partial}_a\) on \(O\) we further find
\begin{equation}
\Pi_H\mathcal{L}_{\Upsilon}\Pi_H\Xi = \Pi_H[\bar{Y}^a\bar{\partial}_a, X^b\delta_b] = \bar{Y}^a\bar{\partial}_aX^b\delta_b = 0\,.
\end{equation}
We thus see that \(\Xi\) is indeed spatio-temporal.
\end{proof}
We have now finally shown that symmetries of a Finsler spacetime can uniquely be described by spatio-temporal symmetries of the induced Cartan geometry on its observer space. This result proves that our definition of symmetries of a Cartan geometry in section~\ref{subsec:cartansym}, together with our definition of spatio-temporal vector fields in section~\ref{subsec:obssym}, provides a suitable generalization of the notion of symmetry also for geometries which go beyond the classical picture of spacetime geometry.

\section{Conclusion}\label{sec:conclusion}
In this article we have introduced three newly developed notions related to Cartan geometry and symmetry. In particular, these notions are: 1. the symmetry, i.e., invariance under diffeomorphisms generated by a vector field on the base manifold \(M\), of a first-order reductive Cartan geometry \((\pi: P \to M, A)\); 2. Cartan geometry on a general space of physical observers, characterized by their positions and velocities; 3. vector fields on observer space generating diffeomorphism which can be interpreted by an observer as transformations of an underlying, possibly local or relative spacetime. We have proven a number of properties of these newly defined notions and their relations to other physically relevant geometries, thereby obtaining the following results.

First, we have discussed our notion of symmetry of first-order reductive Cartan geometries. Our definition essentially differs from the standard definition of (infinitesimal) automorphisms of Cartan geometries by the fact that the latter are defined in terms of (infinitesimal) diffeomorphisms of the total space \(P\), whereas we consider (infinitesimal) diffeomorphisms of the base manifold \(M\). We have applied our definition to Cartan geometries based on two different model Klein geometries, namely the affine and the orthogonal model geometries. We have proven that the standard notions of symmetry of affine, Riemann-Cartan and Riemannian geometries, all of which can be formulated as first-order reductive Cartan geometries and have physical applications as geometries of spacetime, fully agree with our newly defined notion, due to the fact that our notion agrees with the invariance of the linear connections defined by those geometries. The latter also holds true for Weizenb\"ock geometry, where we have found the interesting consequence that the corresponding notion of symmetry is not invariant under local Lorentz transformations of the tetrad field.

Further, we have used our newly introduced definition of observer space Cartan geometries as a starting point to prove that any such geometry \((\pi: P \to O, A)\) induces decompositions of the tangent bundles \(TO\) and \(TP\) into constant rank subbundles, whose elements have a physical interpretation as infinitesimal translations in time and space, boosts and rotations of an observer's frame, generalizing results obtained in~\cite{Gielen:2012fz}. We have then proven that the subbundles of \(TO\) corresponding to boosts and spatial translations are related by a distinguished vector bundle isomorphism, which is unique up to a constant scale factor. These results allowed us to single out a class of vector fields on \(O\), whose physical interpretation is such that they preserve a certain notion of local or relative spacetime, and which we therefore call ``spatio-temporal''.

We have finally applied our definitions and results to Finsler spacetimes, which give rise to observer space Cartan geometries as proven in~\cite{Hohmann:2013fca}. As the main result of this application we have proven that there is a one-to-one correspondence between spacetime vector fields generating symmetries of the Finsler geometry and symmetries of the observer space Cartan geometry in the sense of our newly introduced definition, and that the vector fields generating the latter are automatically spatio-temporal.

In addition to being applicable to specific geometries giving rise to Cartan geometries, our constructions and findings presented in this article also have a number of potential future applications for gravity theories based on Cartan geometry itself. Our definition allows us to construct general Cartan spacetimes with a given symmetry, such as spherical or cosmological symmetry, and in particular symmetric solutions to a given gravity theory. In practice, this amounts to choosing a model geometry \(G/H\) and a principal $H$-bundle \(\pi: \mathcal{P} \to M\), and to determining all $\mathfrak{g}$-valued one-forms on \(\mathcal{P}\) satisfying the conditions on a Cartan connection given in definition~\ref{def:cartan}, which in addition satisfy the invariance conditions given in definition~\ref{def:invcartgeom} for a given set of vector fields on the base manifold \(M\). One may even think of a classification of Cartan spacetimes by their symmetries along the lines of~\cite{Stephani:2003tm}.

Even more interestingly, our findings allow us to go beyond the classical paradigm of an absolute spacetime, and consider the case that spacetime becomes relative to an observer instead. In this situation, in which gravity may be modeled by Cartan geometry on the space of observers, we can essentially apply the same notion of symmetry, augmented by an additional constraint on the symmetry generating vector fields which ensures that they preserve the relative spacetime. Our constructions therefore provide a possible answer to the question how to define solutions to gravity theories on observer space with a certain (such as spherical or cosmological) symmetry, and make it possible to find such solutions.

Apart from being relevant for gravity theories based on these geometries, there are also potential future applications for matter field theories, which use Cartan geometry or any other geometry mentioned above as a background. Here one may naturally ask the question what are the Noether charges corresponding to the notion of symmetry we defined here. Following Noether's theorem, these can be derived by considering an action for matter fields on a Cartan geometric background possessing suitable symmetry. In order to construct suitable matter field actions, the Cartan gauge symmetry of the underlying geometrical background may serve as a guiding principle. We leave the investigation of such matter couplings, following similar considerations as in~\cite{Hehl:1994ue} for metric affine gauge theory, for future work.

\acknowledgments
The author is happy to thank Friedrich Wilhelm Hehl for helpful comments. He gratefully acknowledges the full financial support of the Estonian Research Council through the Postdoctoral Research Grant ERMOS115 and the Startup Research Grant PUT790.

\end{document}